\documentclass[11pt,fleqn,a4paper]{article}  

\usepackage{amsmath,amssymb,amsthm,enumerate,comment}
\usepackage{cancel}
\usepackage[mathscr]{eucal}
\usepackage{multirow}

\usepackage{hyperref}
\hypersetup{colorlinks, linkcolor=blue, citecolor=blue, urlcolor=blue}

\newcommand{\p}{\partial}
\newcommand{\sgn}{\mathop{\rm sgn}\nolimits}
\newcommand{\const}{{\rm const}}

\newcommand{\pr}{\mathop{\rm pr}\nolimits}
\newcommand{\e}{\varepsilon}

\newtheorem{theorem}{Theorem}
\newtheorem{lemma}[theorem]{Lemma}

\newtheorem{proposition}[theorem]{Proposition}
\newtheorem*{proposition*}{Proposition}
{\theoremstyle{definition}

\newtheorem{remark}[theorem]{Remark}

}

\newcommand{\todo}[1][\null]{\ensuremath{\clubsuit}}

\newcommand{\noprint}[1]{}

\newcounter{tbn}

\newcounter{mcasenum}

\begin{document}
\par\noindent {\LARGE\bf
Enhanced group classification\\ of nonlinear diffusion--reaction equations\\ with gradient-dependent diffusivity
\par}

{\vspace{5mm}\par\noindent {\large Stanislav Opanasenko$^{\dag\S}$, Vyacheslav Boyko$^\S$ and Roman O.\ Popovych$^{\ddag\S}$
} \par\vspace{3mm}\par}

{\vspace{3mm}\par\noindent {\it
$^{\dag}$Department of Mathematics and Statistics, Memorial University of Newfoundland,\\
$\phantom{^{\dag}}$St.\ John's (NL) A1C 5S7, Canada
}}
{\vspace{3mm}\par\noindent {\it
$^\ddag$Fakult\"at f\"ur Mathematik, Universit\"at Wien, Oskar-Morgenstern-Platz 1, 1090 Wien, Austria
$\phantom{^\ddag}$\& Mathematical Institute, Silesian University in Opava, Na Rybn\'\i{}\v{c}ku 1, 746 01 Opava,\\
$\phantom{^\ddag}$Czech Republic\par

}}
{\vspace{3mm}\par\noindent {\it
$^{\S}$Institute of Mathematics of NAS of Ukraine, 3 Tereshchenkivs'ka Str., 01004 Kyiv, Ukraine
}}

{\vspace{3mm}\par\noindent {\it
\textup{E-mail:} sopanasenko@mun.ca, boyko@imath.kiev.ua, rop@imath.kiev.ua
}\par}

\vspace{9mm}\par\noindent\hspace*{10mm}\parbox{140mm}{\small
We carry out the enhanced group classification of a class
of (1+1)-dimensional nonlinear diffusion--reaction equations with gradient-dependent diffusivity
using the two-step version of the method of furcate splitting.
For simultaneously finding the equivalence groups
of an unnormalized class of differential equations and a collection of its subclasses,
we suggest an optimized version of the direct method.
The optimization includes the preliminary study of admissible transformations within the entire class
and the successive splitting of the corresponding determining equations with respect to arbitrary elements
and their derivatives depending on auxiliary constraints associated with each of required subclasses.
In the course of applying the suggested technique to subclasses of the class under consideration,
we construct, for the first time, a nontrivial example of finite-dimensional effective generalized equivalence group.
Using the method of Lie reduction and the generalized separation of variables,
exact solutions of some equations under consideration are found.
}\par\vspace{5mm}

\noprint{


}

\section{Introduction}

The recent researches in biology showed that diffusion--reaction equations
draw an attention as a prototype model for pattern formation.
The various patterns, such as fronts, spirals, targets etc.,
can be found in various types of diffusion--reaction systems depending on large discrepancies~\cite{Murray2002,Murray2003}.
The latest application of diffusion--reaction processes are connected to process of morphogenesis
as well as can be relevant to animal coats and skin pigmentation~\cite{Grindrod1996}.
Among other applications are ecological invasions, spread of epidemics, tumour growth and wound healing~\cite{Smoller1994}.

The initial purpose of the present paper was the group classification of the class~$\mathcal R$ of (1+1)-dimensional diffusion--reaction equations
with a gradient-dependent diffusivity,
\begin{equation}\label{eq:RDEsClass}
u_t=f(u_x)u_{xx}+g(u),
\end{equation}
where $f=f(u_x)$ and $g=g(u)$ are smooth functions of their arguments with~$f\ne0$.
This problem had been considered in~\cite{ChernihaKingKovalenko2016} with several weaknesses
(see the conclusion of the present paper),
which made necessary to accurately study it once more.
Using the classical method of Lie reduction and other techniques,
we also planned to construct exact solutions of equations from the regular subclass of the class~$\mathcal R$
that are the most interesting from the Lie-symmetry point of view.

According to the formalized definition of a class of differential equations~\cite{Popovych2006,PopovychKunzingerEshraghi2010},
the complete system of auxiliary equations and inequalities for the arbitrary elements~$f$ and~$g$ of the class~$\mathcal R$ is given by
\begin{gather}\label{eq:RDEsClassAuxiliarySystem}
\begin{split}
&f_t=f_x=f_u=f_{u_t}=f_{u_{tt}}=f_{u_{tx}}=f_{u_{xx}}=0,\quad f\ne0,\\
&g_t=g_x=g_{u_t}=g_{u_x}=g_{u_{tt}}=g_{u_{tx}}=g_{u_{xx}}=0.
\end{split}
\end{gather}
The structure of the class~$\mathcal R$ is not nice from the point of view of equivalence transformations and Lie symmetries.
This is why it is convenient to represent this class as a union of four subclasses,
\[
\mathcal R=\mathcal H\cup\mathcal L\cup\mathcal F\cup\mathcal C.
\]

The subclass~$\mathcal H$ of semilinear equations (called ``nonlinear \emph{heat} equations with source'') is singled out by the constraint~$f_{u_x}=0$.
This class includes all linear equations from the class~$\mathcal R$, which additionally satisfy the constraint~$g_{uu}=0$
and are reduced by simple point transformations to the linear heat equation $u_t=u_{xx}$.
The complete group classification of the subclass~$\mathcal H$ was carried out in~\cite{Dorodnitsyn1982}
in the course of the group classification of the wider class of diffusion--reaction equations
of the general form $u_t=f(u)u_{xx}+g(u)$ with $f\ne0$.
See also \cite[pp.~133--136]{Ames1994V1} for an enhanced representation of these results.

The subclass~$\mathcal L$ consists of equations from the class~$\mathcal R$,
where the values of the arbitrary element~$f$ satisfy the constraint~$(u_x{}^{\!2}f)_{u_x}=0$.
The subclass~$\mathcal L$ is special because each equation from it can be \emph{linearized} to a Kolmogorov equation.
More specifically, the hodograph transformation $\tilde t=t$, $\tilde x=u$, $\tilde u=x$
with~$(\tilde t,\tilde x)$ and~$\tilde u$ being the new independent and dependent variables, respectively,
maps an equation of the form~\eqref{eq:RDEsClass}, where $f=cu_x^{-2}$ with $c=\const\ne0$, to the Kolmogorov equation
\begin{gather}\label{eq:RDEsKolmogorovEqs}
\tilde u_{\tilde t}=c\tilde u_{\tilde x\tilde x}-g(\tilde x)\tilde u_{\tilde x}.
\end{gather}
In particular, this means that the subclass~$\mathcal L$
has completely different point-transformation and Lie-symmetry properties in comparison with the other subclasses,
and its group classification reduces to the group classification of the class of Kolmogorov equations,
which was presented, e.g., in~\cite[Corollary~7]{PopovychKunzingerIvanova2008} up to general point equivalence.

The subclass~$\mathcal F$ is singled out by the constraint $g_u=0$.
The singularity of the subclass~$\mathcal F$ is exhibited by properties of its equivalence transformations,
see Section~\ref{sec:RDEEquivalenceTransformations}.
In particular, the subclass~$\mathcal F$ admits an extension of equivalence group
in comparison with the entire class~$\mathcal R$,
and it is mapped by a family of its equivalence transformations to its subclass~$\mathcal F'$ associated with the additional constraint $g=0$.
Thus, the group classification of the subclass~$\mathcal F$ reduces to the group classification of the subclass~$\mathcal F'$.
The latter subclass consists of ``nonlinear \emph{filtration} equations'',
which are of the form~\eqref{eq:RDEsClass} with $f\ne0$ and~$g=0$.
The group classification of the class~$\mathcal F'$ was carried out in~\cite{AkhatovGazizovIbragimov1987,AkhatovGazizovIbragimov1991}.
Lists of inequivalent Lie-symmetry extensions in this class up to its complete equivalence group and up to a proper subgroup of this group
can also be singled out from the corresponding lists for potential diffusion--convection equations,
presented in~\cite{PopovychIvanova2005}, by selecting cases with the zero convection coefficient.

Each of the additional auxiliary constraints associated with subclasses~$\mathcal H$, $\mathcal L$ and~$\mathcal F$
is related to a special case of solving the group classification problem for the class~$\mathcal R$.
This is why we call the \emph{complement}~$\mathcal C$ of $\mathcal H\cup\mathcal L\cup\mathcal F$ in~$\mathcal R$
the \emph{regular subclass} of~$\mathcal R$.
It is associated, as a subclass of~$\mathcal R$, with the system of the inequalities
$f_{u_x}\ne0$, $(u_x{}^{\!2}f)_{u_x}\ne0$ and $g_u\ne0$.

Note that the union of the subclasses~$\mathcal H$, $\mathcal L$ and~$\mathcal F$ is not disjoint
since there are two nonempty intersections among the pairwise intersections of these subclasses,
$\mathcal H\cap\mathcal F$ and $\mathcal L\cap\mathcal F$.
Unfortunately, there is no partition of the class~$\mathcal R$ into subclasses
that is convenient for group classification.
For example, the equivalence group of the subclass $\mathcal F\setminus(\mathcal H\cup\mathcal L)$
is merely a proper subgroup of the equivalence group of the subclass~$\mathcal F$,
which essentially complicates the group classification of $\mathcal F\setminus(\mathcal H\cup\mathcal L)$
in comparison with~$\mathcal F$.

Although the group classifications for the subclasses~$\mathcal H$, $\mathcal L$ and~$\mathcal F$ are in fact known,
the complete exclusion of these subclasses from the consideration,
as this was done in~\cite{ChernihaKingKovalenko2016}, is not natural.
Thus, there are point transformations
mapping equations from the subclass~$\mathcal F$ to equations from the subclass~$\mathcal C$,
and related Lie-symmetry extensions in~$\mathcal F$ are simpler than their counterparts in~$\mathcal C$.
The same claim is true for the subclass pair $(\mathcal L,\mathcal H)$.
Therefore, to solve the group classification problem for the class~$\mathcal R$,
we carry out the group classification of the regular subclass~$\mathcal C$
using the technique of furcate splitting~\cite{IvanovaPopovychSophocleous2010,NikitinPopovych2001,VaneevaPopovychSophocleous2012},
combine the result with the known group classification of the subclass~$\mathcal F$
and supplement it with the additional inequivalent cases of Lie-symmetry extension from the subclasses~$\mathcal H$ and~$\mathcal L$.

As mentioned above, the initial purpose of the present paper was to correctly solve the group classification problem for the class~$\mathcal R$
but it was changed after the careful analysis of admissible and equivalence transformations within this class.
The generalized equivalence group~$\bar G^{\sim}_{\mathcal F}$ of the subclass~$\mathcal F$ turns out to be nontrivial,
and using it as a conditional generalized equivalence group of the class~$\mathcal R$
simplifies the computation in the course of the group classification of this class
and makes it more consistent with the subclass hierarchy considered for the class~$\mathcal R$.
We also constructed an effective generalized equivalence group~$\hat G^{\sim}_{\mathcal F}$ of the subclass~$\mathcal F$,
which is perhaps the most interesting and unexpected result of the paper
since it gives the first example of nontrivial finite-dimensional effective generalized equivalence group in the literature.
Moreover, the group~$\hat G^{\sim}_{\mathcal F}$ is a proper but not normal subgroup of the group~$\bar G^{\sim}_{\mathcal F}$,
and hence it is not a unique effective generalized equivalence group~$\hat G^{\sim}_{\mathcal F}$ the class~$\mathcal F$.
One more interesting feature of the class~$\mathcal F$ is that
its usual equivalence group is contained in no effective generalized equivalence group of this class.

The rest of the paper is organized as follows.
In Section~\ref{sec:RDEEquivalenceTransformations} we simultaneously compute the equivalence groups
of the entire class~$\mathcal R$ and of the above subclasses of this class
using an original optimized version of the direct method.
The classification of Lie symmetries of equations from the class~$\mathcal R$
is presented in~Section~\ref{sec:RDEsClassGroupClassification}.
Section~\ref{sec:RDEsFurcateSplit} contains the first comprehensive description of the method of furcate splitting.
The two-step version of this method is used in Section~\ref{sec:RDEsGroupClassificationOfC} to solve,
as a part of the group classification problem for the entire class~$\mathcal R$,
the group classification problem for its regular subclass~$\mathcal C$.
In Section~\ref{sec:RDEsSolutions}, we select the three most interesting cases
in the classification list obtained for the subclass~$\mathcal C$
and construct exact solutions of the related equations using Lie reduction or the generalized separation of variables.
Results of the paper are discussed in Section~\ref{sec:RDEsConclusions}.

\section{Equivalence transformations}\label{sec:RDEEquivalenceTransformations}

For simultaneously finding the equivalence groups of an unnormalized class 
of differential equations and a collection of its subclasses,
we suggest an optimized version of the direct method,
which involves the preliminary study of admissible transformations within the entire class
and the successive splitting of the determining equations for these transformations
with respect to the corresponding arbitrary elements and their derivatives,
depending on auxiliary constraints associated with each of required subclasses.

By definition, equivalence transformations for the class~$\mathcal R$ and its subclasses
are point transformations in the joint space of the independent variables~$(t,x)$,
the dependent variable~$u$, its first- and second-order derivatives and the arbitrary elements $f$ and~$g$.
Due to specific form of the arbitrary elements $f$ and~$g$, these transformations
can be defined on spaces with a smaller number of coordinates; cf.~\cite{OpanasenkoBihloPopovych2017}.
Since the arbitrary element~$f$ depends on~$u_x$,
this derivative should be among the coordinates of such a space
although the corresponding transformation components can still be computed
from the $x$- and $u$-components using the chain rule.
Due to the evolution form of equations, the derivative~$u_t$ is not involved in the transformation components for~$u_x$
and thus can be excluded from the coordinates of such a space.
As a result, the minimal list of coordinates for a space underlying equivalence transformations
for the classes~$\mathcal R$, $\mathcal L$, $\mathcal F$ and~$\mathcal C$ is $(t,x,u,u_x,f,g)$.
(The $u_x$-component of equivalence transformations was missed in \cite[Theorem~2]{ChernihaKingKovalenko2016}.)
Since the subclass~$\mathcal H$ is associated with the constraint $f_{u_x}=0$,
the coordinate $u_x$ can be neglected when defining equivalence transformations of this subclass
but for unification we will not use this possibility.
At the same time, in view of the constraint $g=0$
it is convenient to exclude the $g$-component when defining equivalence transformations of the subclass~$\mathcal F'$.
\looseness=-1

In order to compute the equivalence groups of the class~$\mathcal R$ and its subclasses in a uniform way,
we begin this computation with the preliminary study of admissible transformations in this class.
Since the class~$\mathcal R$ consists of (1+1)-dimensional second-order evolution equations
whose right hand sides are affine in the derivative~$u_{xx}$,
any contact admissible transformation in this class
is a prolongation of a point admissible transformation~\cite{PopovychSamoilenko2008}.
Moreover, the $t$-component of any admissible transformation within~$\mathcal R$ depends only on~$t$;
see the related assertions in~\cite{Magadeev1993} and~\cite{Kingston1991} for contact and point transformations
between (1+1)-dimensional evolution equations, respectively.
Therefore, we consider a point transformation of independent and dependent variables of the~form\looseness=-1
\begin{gather}\label{eq:RDEsGenPointTrans}
\tilde t=T(t),\quad \tilde x=X(t,x,u),\quad\tilde u=U(t,x,u),
\end{gather}
where $T_t(X_xU_u-X_uU_x)\ne0$,
that connects the source and the target equations from the class~$\mathcal R$,
\begin{gather}\label{eq:RDEsSource&TargetEqs}
u_t=f(u_x)u_{xx}+g(u)\quad\text{and}\quad
\tilde u_{\tilde t}=\tilde f(\tilde u_{\tilde x})\tilde u_{\tilde x\tilde x}+\tilde g(\tilde u).
\end{gather}
Proceeding with the direct method of finding the equivalence groupoid of a class of equations,
we write the differentiation operators~$\p_{\tilde t}$ and~$\p_{\tilde x}$
with respect to new independent variables in terms of old ones as
\[
\p_{\tilde t}=\frac1{T_t}\left(\mathrm D_t-\frac{\mathrm D_tX}{\mathrm D_xX}\mathrm D_x\right),
\quad \p_{\tilde x}=\frac1{\mathrm D_xX}\mathrm D_x,
\]
where ${\rm D}_t=\p_t+u_t\p_u+u_{tt}\p_{u_t}+u_{tx}\p_{u_x}+\cdots$ and
${\rm D}_x=\p_x+u_x\p_u+u_{tx}\p_{u_t}+u_{xx}\p_{u_x}+\cdots$
are the total derivative operators with respect to~$t$ and~$x$, respectively.
We substitute the expressions for~$\tilde u$, $\tilde u_{\tilde t}$, $\tilde u_{\tilde x}$ and~$\tilde u_{\tilde x\tilde x}$
in terms of old variables into the target equation,
\[
\tilde u_{\tilde x}=V:=\frac{\mathrm D_xU}{\mathrm D_xX}, \quad
\tilde u_{\tilde t}=\frac1{T_t}\left(\mathrm D_tU-V\mathrm D_tX\right),\quad
\tilde u_{\tilde x\tilde x}=\frac{\mathrm D_xV}{\mathrm D_xX},
\]
replace~$u_t$ by $fu_{xx}+g$ and successively split the equation obtained with respect to~$u_{xx}$.
This results in ``expressions'' for the target arbitrary elements~$\tilde f$ and~$\tilde g$,
\begin{gather}\label{eq:RDEsExprForTargetArbitraryElementF}
\tilde f(V)=\frac{(\mathrm D_xX)^2}{T_t}f(u_x),
\\[.5ex]\label{eq:RDEsExprForTargetArbitraryElementG}
\tilde g(U)=\frac{\Delta}{T_t\mathrm D_xX}g(u)-\frac{\mathrm D_xX}{T_t}(V_x+u_xV_u)f(u_x)+\frac{U_t\mathrm D_xX-X_t\mathrm D_xU}{T_t\mathrm D_xX},
\end{gather}
where $\Delta:=X_xU_u-X_uU_x\ne0$.
Since both the source and target arbitrary-element tuples satisfy the auxiliary system~\eqref{eq:RDEsClassAuxiliarySystem},
the equations~\eqref{eq:RDEsExprForTargetArbitraryElementF} and~\eqref{eq:RDEsExprForTargetArbitraryElementG}
imply further determining equations for admissible transformations in the class~$\mathcal R$.
There are two ways for deriving these determining equations.

The first way is to express the operators~$\p_{\tilde t}$, $\p_{\tilde x}$, $\p_{\tilde u}$ and~$\p_{\tilde u_{\tilde x}}$
in terms of the operators~$\p_t$, $\p_x$, $\p_u$ and $\p_{u_x}$, which act on functions of $(t,x,u,u_x)$,
using the equalities implied by the chain rule for these operators:
\begin{gather*}
\p_t=T_t\p_{\tilde t}+X_t\p_{\tilde x}+U_t\p_{\tilde u}+V_t    \p_{\tilde u_{\tilde x}},
\quad
\begin{array}{l}
\p_x=                 X_x\p_{\tilde x}+U_x\p_{\tilde u}+V_x    \p_{\tilde u_{\tilde x}},\\[1ex]
\p_u=                 X_u\p_{\tilde x}+U_u\p_{\tilde u}+V_u    \p_{\tilde u_{\tilde x}},
\end{array}
\quad
\p_{u_x}=                                               V_{u_x}\p_{\tilde u_{\tilde x}}.
\end{gather*}
Then one acts by the operators~$\p_{\tilde t}$, $\p_{\tilde x}$ and~$\p_{\tilde u}$
on the equation~\eqref{eq:RDEsExprForTargetArbitraryElementF}
and the operators~$\p_{\tilde t}$, $\p_{\tilde x}$ and~$\p_{\tilde u_{\tilde x}}$
on the equation~\eqref{eq:RDEsExprForTargetArbitraryElementG}.
The determining equations derived in this way are appropriate
in order to solve the problem of describing the equivalence groupoid~$\mathcal G^\sim_{\mathcal R}$ of the class~$\mathcal R$.
This problem is reduced to the classification of admissible transformations,
which is similar to but more complicated than
the classification of Lie symmetries in Sections~\ref{sec:RDEsClassGroupClassification} and~\ref{sec:RDEsGroupClassificationOfC}
below.
It is not a subject of the present paper although we will need a partial classification of admissible transformations
for proving Theorem~\ref{thm:RDEGroupClassificationUpToGenPointEquiv}.

We use the other way, which gives more compact determining equations for equivalence transformations.
We solve the equations~\eqref{eq:RDEsExprForTargetArbitraryElementF} and~\eqref{eq:RDEsExprForTargetArbitraryElementG}
with respect to $f$ and~$g$, respectively,
\begin{gather}\label{eq:RDEsExprForSourceArbitraryElementF}
f(u_x)=\frac{T_t}{(\mathrm D_xX)^2}\tilde f(V),
\\[.5ex]\label{eq:RDEsExprForSourceArbitraryElementG}
g(u)=\frac{T_t\mathrm D_xX}{\Delta}\tilde g(U)+\frac{(\mathrm D_xX)^2}{\Delta}(V_x+u_xV_u)f(u_x)
-\frac{U_t\mathrm D_xX-X_t\mathrm D_xU}{\Delta},
\end{gather}
separately differentiate the equation~\eqref{eq:RDEsExprForSourceArbitraryElementF} with respect to~$t$, $x$ and~$u$,
then separately differentiate the equation~\eqref{eq:RDEsExprForSourceArbitraryElementG} with respect to~$t$ and $x$
and successively substitute for~$f$ in view of~\eqref{eq:RDEsExprForSourceArbitraryElementF}
as well as substitute the expression~\eqref{eq:RDEsExprForSourceArbitraryElementF} for~$f$
into~\eqref{eq:RDEsExprForSourceArbitraryElementG} and then differentiate the resulting equation with respect to~$u_x$.
This leads to the following classifying equations for admissible transformations within the class~$\mathcal R$:
\begin{gather}\label{eq:RDEsDetEqForAdmTrans1}
\frac{T_tV_z}{(\mathrm D_xX)^2}\tilde f_{\tilde u_{\tilde x}}(V)+\left(\frac{T_t}{(\mathrm D_xX)^2}\right)_z\tilde f(V)=0,\quad z\in\{t,x,u\},
\\[2.5ex]\label{eq:RDEsDetEqForAdmTrans2}
\begin{split}
&\frac{T_t\mathrm D_xX}{\Delta}U_y\tilde g_{\tilde u}(U)+\left(\frac{T_t\mathrm D_xX}{\Delta}\right)_y\tilde g(U)
+\left(\frac{(\mathrm D_xX)^2}{\Delta}(V_x+u_xV_u)\right)_y\frac{T_t}{(\mathrm D_xX)^2}\tilde f(V)
\\
&\qquad-\left(\frac{U_t\mathrm D_xX-X_t\mathrm D_xU}{\Delta}\right)_y=0,\quad y\in\{t,x\},
\end{split}
\\[2ex]\label{eq:RDEsDetEqForAdmTrans3}
\frac{X_u}{\Delta}\tilde g(U)+\frac{V_x+u_xV_u}{(\mathrm D_xX)^2}\tilde f_{\tilde u_{\tilde x}}(V)
+(V_u+V_{xu_x}+u_xV_{uu_x})\frac{\tilde f(V)}{\Delta}
-\frac{U_tX_u-X_tU_u}{T_t\Delta}=0.
\end{gather}
(We divided the last equation by~$T_t$.)
Since the $t$-, $x$- and $u$-components of usual equivalence transformations do not depend on the arbitrary elements~$f$ and~$g$,
in the course of computing the usual equivalence groups of the class~$\mathcal R$ and its subclasses
we can split the classifying equations~\eqref{eq:RDEsDetEqForAdmTrans1}--\eqref{eq:RDEsDetEqForAdmTrans3}
with respect to all parametric derivatives of the target arbitrary elements including~$\tilde f$ and~$\tilde g$ themselves.
At the same time, the complete system of classifying equations for admissible transformations
within a subclass of the class~$\mathcal R$ may be more restrictive
than the equations~\eqref{eq:RDEsDetEqForAdmTrans1}--\eqref{eq:RDEsDetEqForAdmTrans3};
see the proofs of propositions below.

\begin{remark}\label{rem:RDEsEquivTransComponents}
Each equivalence transformation of any subclass of the class~$\mathcal R$
is completely defined by its the $t$-, $x$- and $u$-components.
Indeed, if these components are known,
then the $u_x$-component is computed by the chain rule,
and the expressions for $f$- and $g$-components follow
from the equations~\eqref{eq:RDEsExprForTargetArbitraryElementF} and~\eqref{eq:RDEsExprForTargetArbitraryElementG}.
This is why we do not discuss the derivation of the latter expressions below.
\end{remark}

\begin{proposition}\label{pro:RDEsUsualEquivGroupOfR}
The usual equivalence group~$G^\sim_{\mathcal R}$ of the class~$\mathcal R$
coincides with the usual equivalence groups of its subclasses~$\mathcal H$, $\mathcal L$ and~$\mathcal C$ and
consists of the point transformations in the space with the coordinates $(t,x,u,u_x,f,g)$, whose components are of the form
\begin{gather}\label{eq:RDEsEquivTrans}
\begin{split}&
\tilde t=T_1 t + T_0,\quad
\tilde x=X_1 x + X_0,\quad
\tilde u=U_2 u + U_0,\quad
\tilde u_{\tilde x}=\frac{U_2}{X_1}u_x,\\&
\tilde f=\frac{X_1^{\,2}}{T_1}f,\quad
\tilde g=\frac{U_2}{T_1}g,
\end{split}
\end{gather}
where $T$'s, $X$'s and~$U$'s are arbitrary constants with $T_1X_1U_2\not=0$.
\end{proposition}

\begin{proof}
For each of the classes~$\mathcal R$, $\mathcal H$, $\mathcal L$ and~$\mathcal C$,
the arbitrary element~$\tilde g$ and its derivative~$\tilde g_{\tilde u}$ are not constrained.
Collecting the coefficients of~$\tilde g_{\tilde u}$ and~$\tilde g$
in the equations~\eqref{eq:RDEsDetEqForAdmTrans2} and~\eqref{eq:RDEsDetEqForAdmTrans3}, respectively,
gives the determining equations $U_t=U_x=0$ and $X_u=0$.
Therefore, $X_xU_u\ne0$ and $V=u_xU_u/X_x$.
Then we split the equation~\eqref{eq:RDEsDetEqForAdmTrans2} with $y=t$ with respect to~$\tilde g$ and derive $T_{tt}=0$.
The arbitrary element~$\tilde f$ can also be assumed as a parametric value for splitting,
and its derivative~$\tilde f_{\tilde u_{\tilde x}}$ is either unconstrained or equal to zero or $-2\tilde f/\tilde u_{\tilde x}$.
This is why we can select terms without arbitrary elements and their derivatives in the equation~\eqref{eq:RDEsDetEqForAdmTrans3},
which leads to $X_t=0$.

Now we only need to obtain the equations $X_{xx}=0$ and $U_{uu}=0$,
which is done by considering separately cases with different constraints for~$\tilde f$.
Then
the $t$-, $x$- and $u$-components of usual equivalence transformations have precisely the form~\eqref{eq:RDEsEquivTrans},
and further we follow Remark~\ref{rem:RDEsEquivTransComponents}.
All the constructed transformations preserve
each of the systems of auxiliary constraints for the arbitrary elements
that are associated with the classes~$\mathcal R$, $\mathcal H$, $\mathcal L$ and~$\mathcal C$.

In particular, the values of~$\tilde f$ and~$\tilde f_{\tilde u_{\tilde x}}$ are not constrained by equations
in the classes~$\mathcal R$ and~$\mathcal C$.
This allows us to split with respect to these values and get the equations $V_z=0$,
which are expanded, for $z=x$ and~$z=u$ to the requested equations $X_{xx}=0$ and $U_{uu}=0$, respectively.

In view of the constraint $\tilde f_{\tilde u_{\tilde x}}=0$ for the class~$\mathcal H$,
the equation~\eqref{eq:RDEsDetEqForAdmTrans1} with $z=x$ and the equation~\eqref{eq:RDEsDetEqForAdmTrans3}
respectively reduce to $X_{xx}=0$ and $V_u+V_{xu_x}+u_xV_{uu_x}=0$.
The expansion of the last equation implies $U_{uu}=0$.

Since $\tilde f_{\tilde u_{\tilde x}}=-2\tilde f/\tilde u_{\tilde x}\ne0$ within the class~$\mathcal L$,
the equation~\eqref{eq:RDEsDetEqForAdmTrans1} with $z=u$ is then equivalent to~$V_u=0$, i.e., $U_{uu}=0$,
and the equation~\eqref{eq:RDEsDetEqForAdmTrans3} reduces to $X_{xx}=0$.
\end{proof}

Solving the additional auxiliary equation~$(u_x{}^{\!2}f)_{u_x}=0$ for the arbitrary element~$f$ in the subclass~$\mathcal L$,
we obtain the representation $f=cu_x^{-2}$ with an arbitrary nonzero constant~$c$.
If we reparameterize the subclass~$\mathcal L$ by taking this constant as a new arbitrary element instead of~$f$,
then the corresponding transformation component is $\tilde c=U_2^{\,2}T_1^{-1}c$.

\begin{proposition}\label{pro:RDEsUsualEquivGroupOfF}
The usual equivalence group~$G^\sim_{\mathcal F}$ of the class~$\mathcal F$
is constituted by the point transformations in the space with the coordinates $(t,x,u,u_x,f,g)$, whose components are of the~form
\begin{gather}\label{eq:RDEsEquivTransOfF}
\begin{split}
&\tilde t=T_1t+T_0,\quad
 \tilde x=X_1x+X_0,\quad
 \tilde u=U_1x+U_2u+U_3t+U_0,\quad
 \tilde u_{\tilde x}=\frac{U_1+U_2u_x}{X_1},\\
&\tilde f=\frac{(X_1)^2}{T_1}f,\quad
 \tilde g=\frac{U_2}{T_1}g+\frac{U_3}{T_1},
\end{split}
\end{gather}
where $T$'s, $X$'s and~$U$'s are arbitrary constants with $T_1X_1U_2\ne0$.
\end{proposition}

\begin{proof}
For the class~$\mathcal F$ we should extend the system
of the classifying equations~\eqref{eq:RDEsDetEqForAdmTrans1}--\eqref{eq:RDEsDetEqForAdmTrans3}
with one more equation by replacing the subscript~$y$ by~$z$ in the equation~\eqref{eq:RDEsDetEqForAdmTrans2},
which takes into account the additional auxiliary constraint $g_u=0$ of this class.
Then we substitute $\tilde g_{\tilde u}=0$ into the extended system
and split it with respect to the varying values~$\tilde g$, $\tilde f$ and~$\tilde f_{\tilde u_{\tilde x}}$.
Thus, vanishing the coefficient of~$\tilde g$
and the term without the varying values in the equation~\eqref{eq:RDEsDetEqForAdmTrans3}
results in the equations $X_u=0$ (and hence $X_xU_u\ne0$) and $X_t=0$.
From the equation~\eqref{eq:RDEsDetEqForAdmTrans1} we derive $V_z=0$ and $T_{tt}=0$.
We successively expand the equations $V_z=0$ for $z=t$, $z=u$ and~$z=x$ and split them with respect to~$u_x$,
obtaining 
\[
U_{tx}=U_{tu}=0,\quad U_{xu}=U_{uu}=0\quad\mbox{and}\quad X_{xx}=U_{xx}=0,
\]
respectively.
Then terms in the equation~\eqref{eq:RDEsDetEqForAdmTrans2} without the varying values merely give $U_{tt}=0$.
The obtained equations for the transformations components constitute
the complete system of determining equations for usual equivalence transformations of the class~$\mathcal F$
since the extended version of the system~\eqref{eq:RDEsDetEqForAdmTrans1}--\eqref{eq:RDEsDetEqForAdmTrans3}
is identically satisfied in view of the collection of these equations.
Therefore, the components of all transformations from the group~$G^\sim_{\mathcal F}$ are of the form~\eqref{eq:RDEsEquivTransOfF}
and each transformation whose components are of the form~\eqref{eq:RDEsEquivTransOfF} belongs to this group.
\end{proof}

The group~$G^\sim_{\mathcal F}$ is a nontrivial conditional usual equivalence group of the class~$\mathcal R$
under the condition $g_u=0$
since the usual equivalence group~$G^\sim_{\mathcal R}$ of the class~$\mathcal R$
is a proper subgroup of~$G^\sim_{\mathcal F}$,
which is singled out by the constraints $U_1=U_3=0$ for group parameters.
Elements of~$G^\sim_{\mathcal F}$ with $(U_1,U_3)\ne(0,0)$ are purely conditional equivalence transformations
for the class~$\mathcal R$ under the condition $g_u=0$.

The family of transformations from~$G^\sim_{\mathcal F}$
with $(t,x,u)$-components $\tilde t=t$, $\tilde x=x$, $\tilde u=u-gt$,
which is parameterized by the arbitrary element~$g$,
maps the class~$\mathcal F$ onto its subclass~$\mathcal F'$ singled out by the constraint~$g=0$.
The usual equivalence group of the class~$\mathcal F'$
was found in~\cite{AkhatovGazizovIbragimov1987,AkhatovGazizovIbragimov1991}.
We can easily prove this result using the classifying equations~\eqref{eq:RDEsDetEqForAdmTrans1}--\eqref{eq:RDEsDetEqForAdmTrans3}.

\begin{proposition}\label{pro:RDEsUsualEquivGroupOfF'}
The usual equivalence group~$G^\sim_{\mathcal F'}$ of the class~$\mathcal F'$
consists of the point transformations in the space with the coordinates $(t,x,u,u_x,f)$, whose components are of the form
\begin{gather*}
\tilde t=T_1t+T_0,\quad
\tilde x=X_1x+X_2u+X_0,\quad
\tilde u=U_1x+U_2u+U_0,\quad
\tilde u_{\tilde x}=\frac{U_1+U_2u_x}{X_1+X_2u_x},\\
\tilde f=\frac{(X_1+X_2u_x)^2}{T_1}f,
\end{gather*}
where $T$'s, $X$'s and~$U$'s are arbitrary constants with $T_1(X_1U_2-X_2U_1)\ne0$.
\end{proposition}

\begin{proof}
Restricted to the subclass~$\mathcal F'$ by the substitution $g=0$ and $\tilde g=0$,
the determining equation~\eqref{eq:RDEsExprForSourceArbitraryElementG} itself becomes
classifying for admissible transformations in this subclass.
Therefore, it replaces its differential consequences,
including the equations~\eqref{eq:RDEsDetEqForAdmTrans2} and~\eqref{eq:RDEsDetEqForAdmTrans3}.
The splitting of the equations~\eqref{eq:RDEsExprForSourceArbitraryElementG} and~\eqref{eq:RDEsDetEqForAdmTrans1}
with respect~$\tilde f$ and~$\tilde f_{\tilde u_{\tilde x}}$ merely implies the equations
\begin{gather}\label{eq:RDEsEquivTransOfF'}
V_z=0,\quad
\left(\frac{T_t}{(\mathrm D_xX)^2}\right)_z=0,\quad z\in\{t,x,u\},\quad
U_t\mathrm D_xX-X_t\mathrm D_xU=0,
\end{gather}
which can be further split with respect to~$u_x$.
The second and the first equations of~\eqref{eq:RDEsEquivTransOfF'}
with $z\in\{x,u\}$ successively imply the equations
$X_{xx}=X_{ux}=X_{uu}=0$ and $U_{xx}=U_{ux}=U_{uu}=0$.
The last equation of~\eqref{eq:RDEsEquivTransOfF'} splits into the system $X_zU_t-U_zX_t=0$ with $z\in\{x,u\}$,
which has, as a homogeneous nondegenerate linear system of algebraic equations
with respect to~$(X_t,U_t)$, the zero solution only, i.e., $X_t=U_t=0$.
Then the second equation with $z=t$ is equivalent to $T_{tt}=0$.
The derived equations exhaustively define the $(t,x,u)$-components of equivalence transformations
of the class~$\mathcal F'$.
In view of Remark~\ref{rem:RDEsEquivTransComponents}, this completes the~proof.
\end{proof}

The group~$G^\sim_{\mathcal F'}$ is a nontrivial conditional usual equivalence group
of both the classes~$\mathcal R$ and~$\mathcal F$ under the condition $g=0$.
Elements of~$G^\sim_{\mathcal F'}$ with $X_2\ne0$ have no counterparts
in~$G^\sim_{\mathcal R}$ and~$G^\sim_{\mathcal F}$
and hence they are purely conditional equivalence transformations
for the classes~$\mathcal R$ and~$\mathcal F$ under the condition $g=0$.

The subgroup of~$G^\sim_{\mathcal F}$ preserving the subclass~$\mathcal F'$ of~$\mathcal F$
is associated with constraint $U_3=0$,
and the projection to the space with the coordinates $(t,x,u,u_x,f)$
maps this subgroup to a proper subgroup of~$G^\sim_{\mathcal F'}$.
This is why the group classification of the class~$\mathcal F$ up to $G^\sim_{\mathcal F}$-equivalence
does not reduce to the group classification of the class~$\mathcal F'$ up to $G^\sim_{\mathcal F'}$-equivalence
under the above map of the class~$\mathcal F$ onto its subclass~$\mathcal F'$.
To make the group classifications of the classes~$\mathcal F$ and~$\mathcal F'$ consistent,
we should consider the stronger equivalence in the class~$\mathcal F$
that is associated with generalized equivalence group of this class.

\begin{proposition}\label{pro:RDEsGenEquivGroupOfF}
The generalized equivalence group~$\bar G^\sim_{\mathcal F}$ of the class~$\mathcal F$ is constituted
by the point transformations in the space with the coordinates $(t,x,u,u_x,f,g)$, whose components are of the~form
\begin{gather*}
\tilde t=\bar T^1t+\bar T^0,\quad
\tilde x=\bar X^1x+\bar X^2u-g\bar X^2t+\bar X^0,\quad
\tilde u=\bar U^1x+\bar U^2u+(\bar T^1\bar F-g\bar U^2)t+\bar U^0,\\
\tilde u_{\tilde x}=\frac{\bar U^1+\bar U^2u_x}{\bar X^1+\bar X^2u_x},\quad
\tilde f=\frac{(\bar X^1+\bar X^2u_x)^2}{\bar T^1}f,\quad
\tilde g=\bar F, 
\end{gather*}
where $\bar T$'s, $\bar X$'s, $\bar U$'s and~$\bar F$ are arbitrary smooth functions of~$g$ with $\bar T^1(\bar X^1\bar U^2-\bar X^2\bar U^1)\bar F_g\ne0$.
\end{proposition}

\begin{proof}
By the definition of generalized equivalence transformations
\cite{Meleshko1994,Meleshko1996,OpanasenkoBihloPopovych2017,PopovychKunzingerEshraghi2010},
their $t$-, $x$- and $u$-components may in general depend on arbitrary elements.
At the same time, these components of transformations from~$\bar G^\sim_{\mathcal F}$
can depend only on~$g$ by the following reasons.
The arbitrary element~$f$ depends on~$u_x$.
Each generalized equivalence transformation generates a family of admissible transformations
parameterized by the arbitrary elements~\cite{OpanasenkoBihloPopovych2017}.
Each contact admissible transformation in the class~$\mathcal F$
is the prolongation of a point admissible transformation to the first-order derivatives of~$u$,
i.e., the class~$\mathcal F$ possesses no nontrivial contact admissible transformations.

We follow the proof of Proposition~\ref{pro:RDEsUsualEquivGroupOfF} but,
due to potential dependence of the $t$-, $x$- and $u$-components of generalized equivalence transformations on~$g$,
here we cannot split with respect to this arbitrary element.
This is the only essential difference with the proof of Proposition~\ref{pro:RDEsUsualEquivGroupOfF}.
Similarly to the proof of Proposition~\ref{pro:RDEsUsualEquivGroupOfF'},
the splitting of the equation~\eqref{eq:RDEsDetEqForAdmTrans1}
with respect~$\tilde f$ and~$\tilde f_{\tilde u_{\tilde x}}$ merely implies the equations
\begin{gather}\label{eq:RDEsDetEqsEquivTransOfF'1}
V_z=0,\quad
\left(\frac{T_t}{(\mathrm D_xX)^2}\right)_z=0,\quad z\in\{t,x,u\}.
\end{gather}
The further splitting of the second and the first equations of~\eqref{eq:RDEsDetEqsEquivTransOfF'1} with $z\in\{x,u\}$
with respect to~$u_x$ successively yields the equations
$X_{xx}=X_{ux}=X_{uu}=0$ and $U_{xx}=U_{ux}=U_{uu}=0$.
Thus, $\Delta_x=\Delta_u=0$.
Then the equation~\eqref{eq:RDEsExprForSourceArbitraryElementG} reduces to
$g\Delta+U_t\mathrm D_xX-X_t\mathrm D_xU=\tilde gT_t\mathrm D_xX$.
Differentiating it with respect to~$x$ and~$u$ and splitting with respect to~$u_x$,
we obtain two systems, \mbox{$X_zU_{xt}-U_zX_{xt}=0$} and $X_zU_{ut}-U_zX_{ut}=0$, where $z\in\{x,u\}$,
which are equivalent, in view of the condition $\Delta\ne0$, to
the equations $X_{xt}=U_{xt}=0$ and $X_{ut}=U_{ut}=0$, respectively.
Hence $\Delta_t=0$, and differentiating the equation~\eqref{eq:RDEsExprForSourceArbitraryElementG}
with respect to~$t$ gives $X_zU_{tt}-U_zX_{tt}=0$ with $z\in\{x,u\}$,
i.e., we also have $X_{tt}=U_{tt}=0$.
Taking into account the derived equations $X_{xt}=X_{ut}=0$
in the second equation of~\eqref{eq:RDEsDetEqsEquivTransOfF'1} with $z=t$
immediately gives $T_{tt}=0$.

In view of the constructed equations for admissible transformations,
the $(t,x,u)$-components of generalized equivalence transformations within the class~$\mathcal F$
should be of the form
\[
\tilde t=\bar T^1t+\bar T^0,\quad
\tilde x=\bar X^1x+\bar X^2u+\bar X^3t+\bar X^0,\quad
\tilde u=\bar U^1x+\bar U^2u+\bar U^3t+\bar U^0,
\]
where $\bar T$'s, $\bar X$'s and~$\bar U$'s are smooth functions of~$g$
with $\bar T^1\ne0$ and $\bar X^1\bar U^2-\bar X^2\bar U^1=\Delta\ne0$.
We represent the result of the splitting of the equation~\eqref{eq:RDEsExprForSourceArbitraryElementG}
with respect to~$u_x$ as the system
\begin{gather}\label{eq:RDEsDetEqsEquivTransOfF'Final}
\bar X^1\bar U^3-\bar U^1\bar X^3=\bar X^1T_t\tilde g-g\Delta, \quad
\bar X^2\bar U^3-\bar U^2\bar X^3=\bar X^2T_t\tilde g.
\end{gather}
It is clear from this system that the $g$-component of any generalized equivalence transformation within the class~$\mathcal F$
is a function~$\bar F$ of~$g$ only, and we can choose this function
as a parameter function merely constrained by the inequality $\bar F_g\ne 0$,
which is needed for the transformation nondegeneracy.
Then the system~\eqref{eq:RDEsDetEqsEquivTransOfF'Final} considered as a linear system of algebraic equations
with respect to $(\bar X^3,\bar U^3)$ possesses a unique solution resulting in the form of transformations
from proposition's statement.
\end{proof}

The usual equivalence group~$G^\sim_{\mathcal F}$ is a (finite-dimensional) subgroup
of the generalized equivalence group~\smash{$\bar G^\sim_{\mathcal F}$}
that is singled out from~\smash{$\bar G^\sim_{\mathcal F}$}
by the following system of constraints for the group parameters:
\[
\bar T^0_g=\bar T^1_g=0,\quad
\bar X^0_g=\bar X^1_g=0,\quad \bar X^2=0,\quad
\bar U^0_g=\bar U^1_g=\bar U^2_g=0,\quad \bar T^1\bar F_g=\bar U^2.
\]
Denote by~$\mathcal G^\sim_{\mathcal F}$ the equivalence groupoid of the class~$\mathcal F$
and by~$\mathcal S^\sim_{\mathcal F}$ the subgroupoid of~$\mathcal G^\sim_{\mathcal F}$
generated by the generalized equivalence group~\smash{$\bar G^\sim_{\mathcal F}$}.
The subgroupoid of~$\mathcal G^\sim_{\mathcal F}$ generated by the usual equivalence group~$G^\sim_{\mathcal F}$
is a proper subgroupoid  of~$\mathcal S^\sim_{\mathcal F}$.
Hence the group~\smash{$\bar G^\sim_{\mathcal F}$} is an example of a nontrivial generalized equivalence group,
and it is also a nontrivial conditional generalized equivalence group of the class~$\mathcal R$,
where the specializing attribute ``nontrivial'' is related to both the attributes ``conditional'' and ``generalized''.
The dependence of group parameters on~$g$ is needless for generating admissible transformations in the class~$\mathcal F$
and is merely a manifestation of the fact that the arbitrary element~$g$ is constant within the subclass~$\mathcal F$.
This is why we need to consider an effective generalized equivalence group of the class~$\mathcal F$,
which is a minimal subgroup of~\smash{$\bar G^\sim_{\mathcal F}$}
generating the subgroupoid~$\mathcal S^\sim_{\mathcal F}$ of~$\mathcal G^\sim_{\mathcal F}$.
See~\cite{OpanasenkoBihloPopovych2017} for related definitions.
The only dependence on~$g$ that is essential for generalized equivalence
is the explicit involvement of~$g$ in the $t$-coefficient of the $x$-component.
At the same time, setting the group parameters $\bar T$'s, $\bar X$'s, $\bar U$'s and~$\bar T^1\bar F-\bar U^2$ to be constants
singles out the subset of elements from~\smash{$\bar G^\sim_{\mathcal F}$}
that is not a subgroup of~\smash{$\bar G^\sim_{\mathcal F}$}
although this subset is minimal among subsets of~\smash{$\bar G^\sim_{\mathcal F}$} generating~$\mathcal S^\sim_{\mathcal F}$.
The construction of an effective generalized equivalence group of the class~$\mathcal F$ is in fact more tricky.

\begin{proposition}\label{pro:RDEsEffectiveGenEquivGroupOfF}
An effective generalized equivalence group~$\hat G^\sim_{\mathcal F}$ of the subclass~$\mathcal F$ is constituted by the point transformations
\begin{gather*}
\tilde t=T_1t+T_0,\quad
\tilde x=X_1x+X_2u-X_2gt+X_0,\\
\tilde u=U_1x+U_2u+(1-U_2)gt+U_3t+\frac{T_0}{T_1}g+U_0,\quad
\tilde u_{\tilde x}=\frac{U_1+U_2u_x}{X_1+X_2u_x},\\
\tilde f=\frac{(X_1+X_2u_x)^2}{T_1}f,\quad
\tilde g=\frac{g+U_3}{T_1},
\end{gather*}
where $T$'s, $X$'s and~$U$'s are arbitrary constants with $T_1(X_1U_2-X_2U_1)\ne0$.
\end{proposition}

\begin{proof}
Consider the set~$H_1$ of the point transformations in the space with the coordinates $(t,x,u,u_x,f,g)$, whose components are of the~form
\begin{gather}\label{eq:RDEsSupersetForEffectiveEquivGroup}
\begin{split}
&\tilde t=T_1t+T_0,\\
&\tilde x=X_1x+X_2u+(A_{11}g+A_{10})t+B_{11}g+B_{10},\\
&\tilde u=U_1x+U_2u+(A_{21}g+A_{20})t+B_{21}g+B_{20},\\
&\tilde u_{\tilde x}=\frac{U_1+U_2u_x}{X_1+X_2u_x},\quad
 \tilde f=\frac{(X_1+X_2u_x)^2}{T_1}f,\quad
 \tilde g=\frac{C_1g+C_0}{T_1},
\end{split}
\end{gather}
where $T$'s, $X$'s, $U$'s, $A$'s, $B$'s and $C$'s are arbitrary constants with $T_1(X_1U_2-X_2U_1)C_1\ne0$.
It is obvious that this set is closed with respect to the composition of transformations and taking the inverse,
i.e., it is a (local) transformation group with $\dim H_1=16$.
Then the intersection $H_0:=H_1\cap\bar G^\sim_{\mathcal F}$ of~$H_1$ with~$\bar G^\sim_{\mathcal F}$,
which is singled out from~$H_1$ by the constraints $A_{10}=0$, $A_{11}=-X_2$, $A_{20}=C_0$ and $A_{21}=C_1-U_2$,
is also a group, and $\dim H_0=12$.
The subgroup~$H_0$ of~$\bar G^\sim_{\mathcal F}$ generates the entire subgroupoid~$\mathcal S^\sim_{\mathcal F}$
of~$\mathcal G^\sim_{\mathcal F}$, which is generated by~$\bar G^\sim_{\mathcal F}$.
At the same time, for each fixed pair of the arbitrary elements $(f,g)$,
the subgroupoid~$\mathcal S^\sim_{\mathcal F}$ contains
a precisely nine-parameter family of admissible transformations with the source $(f,g)$.
This is why we should try to find three more constraints for group parameters of the group~$H_1$
in order to construct a nine-dimensional subgroup of~$H_0$ that still generates the entire~$\mathcal S^\sim_{\mathcal F}$.

We analyze the composition of two arbitrary elements from the group~$H_0$,
$\hat{\mathcal T}=\tilde{\mathcal T}\mathcal T$ with $\tilde{\mathcal T},\mathcal T\in H_0$.
These generalized equivalence transformations have the general form~\eqref{eq:RDEsSupersetForEffectiveEquivGroup},
where group parameters satisfy the above constraints for the subgroup~$H_0$.
We additionally reparameterize~$H_0$ with replacing the parameter~$B_{21}$ by $B_{21}'+T_0/T_1$ and
mark the group-parameter values corresponding to~$\hat{\mathcal T}$ and~$\tilde{\mathcal T}$
by hats and tildes, respectively.
We obtain, in particular, the following expressions
for group-parameter values of the composition~$\hat{\mathcal T}$:
\begin{gather*}
\hat C_1=\tilde C_1C_1,\quad
\hat B_{11} =\tilde X_1B_{11}+\tilde X_2B_{21}'+\frac{\tilde B_{11} }{T_1},\quad
\hat B_{21}'=\tilde U_1B_{11}+\tilde U_2B_{21}'+\frac{\tilde B_{21}'}{T_1},
\end{gather*}
which imply that the constrains $C_1=1$, $B_{11}=B_{21}'=0$
singling out~$\hat G^\sim_{\mathcal F}$ from the subgroup~$H_0$
are preserved by the composition of transformations and taking the inverse in~$H_0$.
Therefore, $\hat G^\sim_{\mathcal F}$ is really a group.
It generates the entire subgroupoid~$\mathcal S^\sim_{\mathcal F}$ of~$\mathcal G^\sim_{\mathcal F}$,
and any its proper subset does not possess this property,
i.e., it is a minimal subgroup of~$\bar G^\sim_{\mathcal F}$ with this property.
\end{proof}

The usual equivalence group~\smash{$G^\sim_{\mathcal F}$} of the subclass~$\mathcal F$
is not contained in the effective generalized equivalence group~\smash{$\hat G^\sim_{\mathcal F}$}
constructed in Proposition~\ref{pro:RDEsEffectiveGenEquivGroupOfF}.
The intersection \smash{$G^\sim_{\mathcal F}\cap\hat G^\sim_{\mathcal F}$}
is singled out from~$G^\sim_{\mathcal F}$ by the constraints $T_0=0$ and $U_2=1$.

To prove an assertion generalizing the above claim,
we need to consider the infinitesimal counterparts of related groups.
For convenience, we introduce the following dual notation
for relevant vector fields on the space with the coordinates $(t,x,u,u_x,f,g)$:
\begin{gather*}
Q^1=P^t=\p_t,\quad
Q^2=D^t=t\p_t-f\p_f-g\p_g,\\
Q^3=P^x=\p_x,\quad
Q^4=D^x=x\p_x-u_x\p_{u_x}+2f\p_f,\\
Q^5=P^u=\p_u,\quad
Q^6=D^u=u\p_u+u_x\p_{u_x}+g\p_g,\\
Q^7=Z^t=t\p_u+\p_g,\quad
Q^8=Z^x=x\p_u+\p_{u_x},\quad
Q^9=R=(u-gt)\p_x-u_x^{\,2}\p_{u_x}+2u_xf\p_f.
\end{gather*}
Up to the anticommutativity of the Lie bracket,
the nonzero commutation relations between these vector fields are exhausted by
\begin{gather*}
[P^t,D^t]= P^t,\ [P^x,D^x]= P^x,\ [P^u,D^u]=P^u,\ [P^t,Z^t]=P^u,\ [P^x,Z^x]=P^u,\\
[Z^t,D^t]=-Z^t,\ [Z^x,D^x]=-Z^x,\ [Z^t,D^u]=Z^t,\ [Z^x,D^u]=Z^x,\\
[P^t,R]= -gP^x,\ [P^u,R  ]= P^x,\ [D^x,R  ]= -R,\ [D^u,R  ]=  R,\ [Z^x,R  ]=D^x-D^u+gZ^t.
\end{gather*}

The Lie algebras~$\mathfrak g^\sim_{\mathcal F}$, $\bar{\mathfrak g}^\sim_{\mathcal F}$ and $\hat{\mathfrak g}^\sim_{\mathcal F}$
of the groups~$G^\sim_{\mathcal F}$, $\hat G^\sim_{\mathcal F}$ and $\bar G^\sim_{\mathcal F}$
are naturally called the usual equivalence algebra, the generalized equivalence algebra
and an effective generalized equivalence algebra of the class~$\mathcal F$, respectively.
Each of them is merely the set of infinitesimal generators of one-parameter subgroups of the corresponding group.
In order to construct all such generators,
we successively take one of the group parameter in the respective general form of group elements
to depend on a continuous subgroup parameter~$\delta$ and
set the other parameter-functions to their values corresponding to the identity transformations,
which are $T_1=X_1=U_2=1$ and $T_0=X_0=X_2=U_0=U_1=U_3=0$ for the groups~$G^\sim_{\mathcal F}$ and~$\hat G^\sim_{\mathcal F}$
(the parameter~$X_2$ is relevant only for~$\hat G^\sim_{\mathcal F}$)
and similarly $\bar T^1=\bar X^1=\bar U^2=1$, $\bar T^0=\bar X^0=\bar X^2=\bar U^0=\bar U^1=0$ and $\bar F=g$
for the group~$\bar G^\sim_{\mathcal F}$.
Then we differentiate the transformation components with respect to~$\delta$ and evaluate the result at $\delta=0$.
As a result, we derive that
\begin{gather*}
\mathfrak g^\sim_{\mathcal F}=\langle Q^1,\dots,Q^8\rangle,\quad
\bar{\mathfrak g}^\sim_{\mathcal F}=\left\{\sum_{i=1}^9\vartheta^i(g)Q^i\right\},\\
\hat{\mathfrak g}^\sim_{\mathcal F}=\langle P^t+gP^u,\,D^t,\,P^x,\,D^x,\,P^u,\,D^u-gZ^t,\,Z^t,\,Z^x,\,R\rangle,
\end{gather*}
where the coefficients $\vartheta$'s run through the set of smooth functions of~$g$,
i.e., the algebra~$\bar{\mathfrak g}^\sim_{\mathcal F}$ is the module over the ring of smooth functions of~$g$
with basis $(Q^1,\dots,Q^9)$ equipped with the Lie bracket of vector fields.

\begin{theorem}\label{thm:RDEsOnContainingUsualEquivGroupOfF}
Any effective generalized equivalence group of the class~$\mathcal F$
does not contain the usual equivalence group~$G^\sim_{\mathcal F}$ of this class.
\end{theorem}

\begin{proof}
We prove the following re-formulated assertion:
Suppose that a subgroup of the generalized equivalence group~$\bar G^\sim_{\mathcal F}$ of the class~$\mathcal F$
contains the usual equivalence group~$G^\sim_{\mathcal F}$ of this class
and generates the same subgroupoid of the equivalence groupoid~$\mathcal G^\sim_{\mathcal F}$
as the entire group~$\bar G^\sim_{\mathcal F}$ does.
Then this subgroup is not an effective generalized equivalence group of the class~$\mathcal F$.

A complete list of discrete usual equivalence transformations of the class~$\mathcal F$
that are independent up to combining with each other
and with continuous usual equivalence transformations of this class
is exhausted by the involutions~$I^t$, $I^x$ and~$I^u$
alternating the signs of $(t,f,g)$, $(x,u_x)$ and~$(u,u_x,g)$, respectively.
Among generalized equivalence transformations,
there is one more independent discrete transformation~$I^g$:
$(\tilde t,\tilde x,\tilde u,\tilde u_{\tilde x},\tilde f,\tilde g)=(t,x,u-2gt,u_x,f,-g)$.
Discrete equivalence transformations play an auxiliary role in the course of the proof.

It suffices to prove the infinitesimal counterpart of the above assertion, which states the following.
Let a subalgebra~$\mathfrak h$ of~$\bar{\mathfrak g}^\sim_{\mathcal F}$
contain~$\mathfrak g^\sim_{\mathcal F}$ and a vector field $Q=\sum_{i=1}^9\zeta^iQ^i$,
where $\zeta^i=\zeta^i(g)$ are smooth functions of~$g$ with $\zeta^9\ne0$,
be invariant with respect to discrete transformations in~$G^\sim_{\mathcal F}$,
$I^t_*\mathfrak h,I^x_*\mathfrak h,I^u_*\mathfrak h\subseteq\mathfrak h$,
and be associated with a transformation (pseudo)group.
Then this subalgebra properly contains
another subalgebra~$\mathfrak s$ among whose elements there are $K^j=\sum_{i=1}^9\chi^{ij}Q^i$,
where $\chi^{ij}$, $i,j=1,\dots,9$, are smooth functions of~$g$ with $\det(\chi^{ij})\ne0$,
and which is also invariant with respect to~$I^t_*$, $I^x_*$ and~$I^u_*$
and is associated with a transformation (pseudo)group.
Here the subscript ``*'' combined with the notation of a point transformation
denotes pushing forward vector fields on the same manifold by this transformation.

If the algebra~$\mathfrak h$ contains the pure vector field~$R$,
then we commute~$R$ with elements of~$\mathfrak g^\sim_{\mathcal F}$ and successively obtain that
\[
[R,P^t]=gP^x\in\mathfrak h,\quad
[gP^x,Z^x]=gP^u\in\mathfrak h,\quad
[Z^x,R]=D^t-D^u+gZ^t\in\mathfrak h.
\]
Hence $gZ^t\in\mathfrak h$, i.e.,
$\mathfrak h\supseteq\mathfrak g^\sim_{\mathcal F}+\langle gP^x,\,gP^u,\,gZ^t\rangle\varsupsetneq \hat{\mathfrak g}^\sim_{\mathcal F}$.
We can choose $\mathfrak s=\hat{\mathfrak g}^\sim_{\mathcal F}$.
Then we also have $I^t_*\mathfrak s=I^x_*\mathfrak s=I^u_*\mathfrak s=I^g_*\mathfrak s=\mathfrak s$.

Otherwise, we compute the commutators
\begin{gather*}
[Q,D^x]=\zeta^9R-\zeta^8Z^x+\zeta^3P^x\in\mathfrak h,\\
[\zeta^9R-\zeta^8Z^x+\zeta^3P^x,D^x]=\zeta^9R+\zeta^8Z^x+\zeta^3P^x\in\mathfrak h,\\
[\zeta^9R-\zeta^8Z^x+\zeta^3P^x,D^t+D^u]=-\zeta^9R-\zeta^8Z^x\in\mathfrak h,
\end{gather*}
and thus derive that $\zeta^9R\in\mathfrak h$, and $\zeta^9\ne\const$.
In the same way, we can show that for any element $\sum_{i=1}^9\vartheta^i(g)Q^i\in\mathfrak h$,
the element $\vartheta^3P^x$ and thus the element $\vartheta^3P^u=[\vartheta^3P^x,Z^x]$ also belong to~$\mathfrak h$.
Taking two more commutators,
\begin{gather*}
[Z^x,\zeta^9R]=\zeta^9(D^x-D^u+gZ^t)\in\mathfrak h,\\
[Z^x,\zeta^9(D^x-D^u+gZ^t)]=-2\zeta^9 Z^x\in\mathfrak h,
\end{gather*}
we get $\zeta^9 Z^x\in\mathfrak h$.
Consider the span
\[
\mathfrak s=\langle P^t,\,D^t,\,Z^t,\,D^x,\,D^u,\,\beta P^x,\,\beta P^u,\,\alpha R,\,\alpha(D^x-D^u+gZ^t),\,\alpha Z^x\mid\alpha R,\beta P^x\in\mathfrak h \rangle.
\]
It is a subalgebra of~$\mathfrak h$.
Since the entire algebra~$\mathfrak h$ is invariant with respect to~$I^t_*$, $I^x_*$ and~$I^u_*$
and is associated with a transformation (pseudo)group,
the subalgebra~$\mathfrak s$ has the same properties.
In view of $R\notin\mathfrak h$, the parameter function~$\alpha$ does not take constant values.
Hence $Z^x\notin\mathfrak s$, i.e., $\mathfrak s\subsetneq\mathfrak h$.
As the required elements $K^j$, $j=1,\dots,9$, we can choose
$P^t$, $D^t$, $Z^t$, $D^x$, $D^u$, $P^x$, $P^u$, $\zeta^9R$ and~$\zeta^9Z^x$.

Therefore, the algebra~$\mathfrak h$ is not an effective generalized equivalence algebra of the class~$\mathcal F$.
\end{proof}

\section{Classification of Lie symmetries}\label{sec:RDEsClassGroupClassification}

In order to compute the maximal Lie invariance algebras of equations from the class~$\mathcal R$,
we employ the infinitesimal method~\cite{Olver1993,Ovsiannikov1982}.
The infinitesimal generator of a one-parameter Lie-symmetry group of an equation from the class~$\mathcal R$
is a vector field $Q=\tau(t,x,u)\p_t+\xi(t,x,u)\p_x+\eta(t,x,u)\p_u$ on the space with coordinates~$(t,x,u)$
with the components~$\tau$, $\xi$ and~$\eta$ being the smooth functions of these coordinates.
The infinitesimal invariance criterion requires that
\begin{gather}\label{eq:RDEsInfInvCriterion}
Q^{(2)}(u_t-f(u_x)u_{xx}-g(u))\big|_{u_t=f(u_x)u_{xx}+g(u)}=0,
\end{gather}
where $Q^{(2)}$ is the second prolongation of the vector field~$Q$ defined by the well-known prolongation formula~\cite{Olver1993},
$
 Q^{(2)}=Q+\eta^{(1,0)}\p_{u_t}+\eta^{(0,1)}\p_{u_x}+\eta^{(2,0)}\p_{u_{tt}}+\eta^{(1,1)}\p_{u_{tx}}+\eta^{(0,2)}\p_{u_{xx}},
$
with the coefficients~$\eta$'s defined as
$
\eta^\alpha=\mathrm D^\alpha\left(\eta-\tau u_t-\xi u_x\right)+\tau u_{\alpha+\delta_1}+\xi u_{\alpha+\delta_2}.
$
Here $\alpha=(\alpha_1,\alpha_2)$ is a multi-index,
$\mathrm D^\alpha=\mathrm D_t^{\alpha_1}\mathrm D_x^{\alpha_2}$,
$\mathrm D_t=\p_t+\sum_\alpha u_{\alpha+\delta_1}\p_{u_\alpha}$
and $\mathrm D_x=\p_x+\sum_\alpha u_{\alpha+\delta_2}\p_{u_\alpha}$
are the total derivative operators with respect to~$t$ and~$x$,
respectively, $\delta_1=(1,0)$ and~$\delta_2=(0,1)$.

Since the class~$\mathcal R$ consists of evolution equations,
the $t$-component of~$Q$ does not depend on~$x$ and~$u$, i.e., $\tau=\tau(t)$~\cite{Kingston1991,Magadeev1993}.

We expand the expression in the left hand side of~\eqref{eq:RDEsInfInvCriterion},
substitute $f(u_x)u_{xx}+g(u)$ for~$u_t$ and then split the resulting equation with respect to~$u_{xx}$.
This gives the system of determining equations for the components of Lie-symmetry vector fields of equations from the class~$\mathcal R$,
\begin{subequations}\label{eq:RDEsDetEqSystem}
\begin{gather}\label{eq:RDEsDetEqA}
\left(-\xi_uu_x{}^{\!2}+(\eta_u-\xi_x)u_x+\eta_x\right)f_{u_x}+(-2\xi_uu_x+\tau_t-2\xi_x)f=0,
\\[1ex]\label{eq:RDEsDetEqB}
\begin{split}
&\left(-\xi_{uu}u_x^3+(\eta_{uu}-2\xi_{xu})u_x{}^{\!2}+(2\eta_{xu}-\xi_{xx})u_x+\eta_{xx}\right)f+(\xi_t+\xi_ug)u_x
\\
&\qquad{}+\eta g_u+(\tau_t-\eta_u)g-\eta_t=0.
\end{split}
\end{gather}
\end{subequations}
Thus, the problem of group classification for the class~$\mathcal R$ reduces to
the classification of solutions of the system~\eqref{eq:RDEsDetEqSystem},
depending on values of the arbitrary elements~$f$ and~$g$
up to $G^\sim_{\mathcal R}$-equivalence or up to the general point equivalence.

To find the kernel Lie invariance algebras of the class~$\mathcal R$
and of its subclasses~$\mathcal H$, $\mathcal L$, $\mathcal F$, $\mathcal F'$ and~$\mathcal C$
(each of these algebras is the intersection of the maximal Lie invariance algebras of equations
from the corresponding (sub)class),
we successively split the equations~\eqref{eq:RDEsDetEqA} and~\eqref{eq:RDEsDetEqB}
with respect to the arbitrary elements and their derivatives and with respect to~$u_x$,
taking into account the associated auxiliary equations for arbitrary elements.
See also the papers~\cite{ChernihaKingKovalenko2016}, \cite{Dorodnitsyn1982} and~\cite{AkhatovGazizovIbragimov1987,AkhatovGazizovIbragimov1991}
for the kernel Lie invariance algebras of the classes~$\mathcal R$, $\mathcal H$ and~$\mathcal F'$, respectively.

\begin{proposition}
The kernel Lie invariance algebra~$\mathfrak g^\cap_{\mathcal R}$ of the class~$\mathcal R$
coincides with the kernel Lie invariance algebras~$\mathfrak g^\cap_{\mathcal H}$ and~$\mathfrak g^\cap_{\mathcal C}$
of its subclasses~$\mathcal H$ and~$\mathcal C$,
and it is spanned by the vector fields~$\p_t$ and~$\p_x$,
\[
\mathfrak g^\cap_{\mathcal R}=\mathfrak g^\cap_{\mathcal H}=\mathfrak g^\cap_{\mathcal C}=\langle\p_t,\,\p_x\rangle.
\]
The kernel Lie invariance algebras of the subclasses~$\mathcal L$, $\mathcal F$ and~$\mathcal F'$ are respectively
\[
\mathfrak g^\cap_{\mathcal L}=\langle\p_t,\,\p_x,\,x\p_x\rangle,\quad
\mathfrak g^\cap_{\mathcal F}=\langle\p_t,\,\p_x,\p_u\rangle,\quad
\mathfrak g^\cap_{\mathcal F'}=\langle\p_t,\,\p_x,\,\p_u,\,2t\p_t+x\p_x+u\p_u\rangle.
\]
\end{proposition}

Accurately merging the group classifications of the subclasses~$\mathcal C$, $\mathcal F'$, $\mathcal H$ and $\mathcal L$
that is presented in Lemma~\ref{pro:RDEsGroupClassificationOfC} below,
that were carried out in~\cite{AkhatovGazizovIbragimov1987,AkhatovGazizovIbragimov1991} and in~\cite{Dorodnitsyn1982},
and that can be obtained by mapping with the hodograph transformation $\tilde t=t$, $\tilde x=u$, $\tilde u=x$
from the group classification of the class of Kolmogorov equations given in~\cite{PopovychKunzingerIvanova2008}, respectively,
we get the following assertion.

\begin{theorem}\label{thm:RDEGroupClassification}
A complete list of inequivalent Lie-symmetry extensions
in the class~$\mathcal R$ is exhausted by the cases given in Table~\ref{tab:RDEGroupClassification},
where we use $G^\sim_{\mathcal R}$-, $\bar G^\sim_{\mathcal F}$-, $G^\sim_{\mathcal R}$- and $\mathcal G^\sim_{\mathcal L}$-equivalence
for equations from the subclasses~$\mathcal C$, $\mathcal F$, $\mathcal H$ and~$\mathcal L$, respectively.
\begin{table}[!p]\footnotesize
\caption{Group classification of the class~$\mathcal R$}
\label{tab:RDEGroupClassification}
\begin{center}
\renewcommand{\arraystretch}{1.43}
\tabcolsep=.5ex
\setcounter{tbn}{-1}
\begin{tabular}{|c|c|c@{\hspace{\tabcolsep}\vline\hspace{1.5ex}}l|}
\hline
no.\!\! &     $f(u_x)$     &$g(u)$                  & \hfil Vector fields spanning the maximal Lie invariance algebra\\
\hline
\multicolumn{4}{|c|}{the subclass~$\mathcal C$, \ up to $G^\sim_{\mathcal R}$-equivalence}\\
\hline
\refstepcounter{tbn}\thetbn  \label{RDE:f=any,g=any}         &$\forall$        &$\forall$                 & $\p_t$, $\p_x$\\
\refstepcounter{tbn}\thetbn  \label{RDE:f=any,g=u}           &$\forall$        &$u$                       & $\p_t$, $\p_x$, $e^t\p_u$\\
\refstepcounter{tbn}\thetbn  \label{RDE:f=any,g=1/u}         &$\forall$        &$u^{-1}$                  & $\p_t$, $\p_x$, $2t\p_t+x\p_x+u\p_u$\\
\refstepcounter{tbn}\thetbn  \label{RDE:f=u_x^k,g=e^u}       &$|u_x|^n$        &$\varepsilon e^u$         & $\p_t$, $\p_x$, $(n+2)t\p_t+x\p_x-(n+2)\p_u$\\
\refstepcounter{tbn}\thetbn a\label{RDE:f=u_x^k,g=u^m}       &$|u_x|^n$        &$|u|^m$                   & $\p_t$, $\p_x$, $(1-m)(n+2)t\p_t+(n+1-m)x\p_x+(n+2)u\p_u$\\
\noprint{          }\thetbn b                                &$|u_x|^n$        &$|u|^{n+1}+\varepsilon u$ & $\p_t$, $\p_x$, $e^{-\varepsilon nt}(\varepsilon \p_t+u\p_u)$\\
\refstepcounter{tbn}\thetbn  \label{RDE:f=1/(u_x+1),g=u}     &$(u_x+1)^{-1}$   &$\varepsilon u$           & $\p_t$, $\p_x$, $e^{\varepsilon t}\p_u$, $e^{\varepsilon t}(\p_t+\varepsilon(u+x)\p_u)$\\
\refstepcounter{tbn}\thetbn a\label{RDE:f=u_x,g=u^2}         &$u_x$            &$u^2$                     & $\p_t$, $\p_x$, $t\p_t-u\p_u$, $t^2\p_t-(2tu+1)\p_u$\\
\noprint{          }\thetbn b                                &$u_x$            &$u^2+1$                   & $\p_t$, $\p_x$, $\cos2t\,(\p_t+2\p_u)+2u\sin2t\,\p_u$, $\sin2t\,(\p_t+2\p_u)-2u\cos2t\,\p_u$\\
\noprint{          }\thetbn c                                &$u_x$            &$u^2-1$                   & $\p_t$, $\p_x$, $e^{2t}(\p_t-2(u+1)\p_u)$, $e^{-2t}(\p_t+2(u-1)\p_u)$\\
\refstepcounter{tbn}\thetbn  \label{RDEPeculiarCase}         &$u_x(u_x+1)^{-3}$&$\varepsilon u$           & $\p_t$, $\p_x$, $e^{\varepsilon t}\p_u$, $e^{-\varepsilon t}(\p_t-\varepsilon u\p_x+\varepsilon u\p_u)$\\
\ref{RDE:f=u_x^k,g=0}b                                       &$|u_x|^n$        &$\varepsilon u$           & $\p_t$, $\p_x$, $e^{\varepsilon t}\p_u$, $nx\p_x+(n+2)u\p_u$, $e^{-\varepsilon nt}(\p_t+\varepsilon u\p_u)$\\
\hline
\multicolumn{4}{|c|}{the subclass~$\mathcal F$, \ up to $\bar G^\sim_{\mathcal F}$-equivalence}\\
\hline
\refstepcounter{tbn}\thetbn  \label{RDE:f=any,g=0}           &$\forall$        &0                         & $\p_t$, $\p_x$, $\p_u$, $2t\p_t+x\p_x+u\p_u$\\
\refstepcounter{tbn}\thetbn a\label{RDE:f=u_x^k,g=0}         &$|u_x|^n$        &0                         & $\p_t$, $\p_x$, $\p_u$, $2t\p_t+x\p_x+u\p_u$, $nt\p_t-u\p_u$\\
\refstepcounter{tbn}\thetbn  \label{RDE:f=e^u_x,g=0}         &$e^{u_x}$        &0                         & $\p_t$, $\p_x$, $\p_u$, $2t\p_t+x\p_x+u\p_u$, $t\p_t-x\p_u$\\
\refstepcounter{tbn}\thetbn  \label{RDE:fWithArctan,g=0}     &$\dfrac{e^{m\arctan u_x}}{u_x{}^{\!2}+1}$&0       & $\p_t$, $\p_x$, $\p_u$, $2t\p_t+x\p_x+u\p_u$, $mt\p_t+u\p_x-x\p_u$\\
\refstepcounter{tbn}\thetbn a\label{RDE:f=1,g=0}             &1                &0                         & $\p_t$, $\p_x$, $\p_u$, $2t\p_t+x\p_x$, $u\p_u$, $2t\p_x-xu\p_u$,\\[-.5ex]&&& $4t^2\p_t+4tx\p_x-(x^2+2t)u\p_u$, $h\p_u$
\\
\hline
\multicolumn{4}{|c|}{the subclass~$\mathcal H$, \ up to $G^\sim_{\mathcal R}$-equivalence, \ only cases additional to Case~\ref{RDE:f=1,g=0}a}\\
\hline
\ref{RDE:f=u_x^k,g=e^u}$'$                                  &1                 &$\varepsilon e^u$         & $\p_t$, $\p_x$, $2t\p_t+x\p_x-2\p_u$\\
\ref{RDE:f=u_x^k,g=u^m}a$'$                                 &1                 &$|u|^m$                   & $\p_t$, $\p_x$, $2(1-m)t\p_t+(1-m)x\p_x+2u\p_u$\\
\refstepcounter{tbn}\thetbn   \label{RDE:f=1,g=ulnu}        &1                 &$\varepsilon u\ln|u|$     & $\p_t$, $\p_x$, $e^{\varepsilon t}(2\p_x-\varepsilon xu\p_u)$, $e^{\varepsilon t}u\p_u$\\
\ref{RDE:f=1,g=0}b                                          &1                 &$\varepsilon u$           & $\p_t$, $\p_x$, $2t\p_t+x\p_x+2\varepsilon tu\p_u$, $u\p_u$, $2t\p_x-xu\p_u$,\\[-.5ex]&&& $4t^2\p_t+4tx\p_x-(x^2+2t-4\varepsilon t^2)u\p_u$, $h\p_u$\!\!\!\!\\
\cancel{\ref{RDE:f=1,g=0}c}                                 &1                 &$1$                       & $\p_t$, $\p_x$, $2t\p_t+x\p_x+2t\p_u$, $(u-t)\p_u$, $2t\p_x-x(u-t)\p_u$,\\[-.5ex]&&& $4t^2\p_t+4tx\p_x-\big((x^2{+}2t)(u-t)-4t^2\big)\p_u$, $h\p_u$
\\
\hline
\multicolumn{4}{|c|}{the subclass~$\mathcal L$, \ up to $\mathcal G^\sim_{\mathcal L}$-equivalence}\\
\hline
\refstepcounter{tbn}\thetbn  \label{RDE:f=1/u_x^2,g=any}    &$u_x^{-2}$        &$\forall$                 & $\p_t$, $\tilde h\p_x$, $x\p_x$\\
\refstepcounter{tbn}\thetbn  \label{RDE:f=1/u_x^2,g=1/u}    &$u_x^{-2}$        &$\mu u^{-1}$              & $\p_t$, $\tilde h\p_x$, $x\p_x$, $4t\p_t+2u\p_u$, $4t^2\p_t+4tu\p_u-\big(u^2+2(1+\nu)t\big)x\p_x$\\
\refstepcounter{tbn}\thetbn  \label{RDE:f=1/u_x^2,g=WithTan}&$u_x^{-2}$ &$\dfrac{1-2\nu\tan(\nu\ln|u|)}u$ & $\p_t$, $\tilde h\p_x$, $x\p_x$, $2t\p_t+u\p_u-xug\p_x$, $4t^2\p_t+4tu\p_u-(u^2+2t+2tug)x\p_x$\\
\cancel{\ref{RDE:f=1,g=0}d}                                 &$u_x^{-2}$        &0                         & $\p_t$, $\tilde h\p_x$, $x\p_x$, $2t\p_t+u\p_u$, $4t^2\p_t+4tu\p_u-(u^2+2t)x\p_x$, $2t\p_u-xu\p_x$
\\
\hline
\end{tabular}
\end{center} 
$n\in\mathbb R\setminus\{0,-2\}$, $m\in\mathbb R$.
$\varepsilon\ne0$, $\varepsilon=\pm1\bmod G^\sim_{\mathcal R}$.
$m\ne-1,0,1$ and $(n,m)\ne(1,2)$ in Case~\ref{RDE:f=u_x^k,g=u^m}a.
$m\ne0,1$ in Case~\ref{RDE:f=u_x^k,g=u^m}a$'$.
$n\ne\pm1$ in Case~\ref{RDE:f=u_x^k,g=u^m}b.
$n\geqslant-1\bmod\bar G^\sim_{\mathcal F}$ in Case~\ref{RDE:f=u_x^k,g=0}a.
$m\geqslant0\bmod G^\sim_{\mathcal R}$ in Case~\ref{RDE:fWithArctan,g=0}.
In~Cases~\ref{RDE:f=1,g=0}a, \ref{RDE:f=1,g=0}b and \ref{RDE:f=1,g=0}c,
the parameter function~$h=h(t,x)$ runs through the solution set of the corresponding equation,
$h_t=h_{xx}$, $h_t=h_{xx}+\varepsilon h$ and $h_t=h_{xx}+\varepsilon$, respectively.
In~Cases~\ref{RDE:f=1/u_x^2,g=any}--\ref{RDE:f=1/u_x^2,g=WithTan} and~\ref{RDE:f=1,g=0}d,
the real constant parameters~$\mu$ and~$\nu$ satisfy, up to point transformations, the constraints $\mu\geqslant1$, $\mu\ne 2$ and $\nu>0$,
and the parameter function~$\tilde h=\tilde h(t,u)$ runs through the solution set of
the corresponding linear equation, $\tilde h_t=\tilde h_{uu}-g(u)\tilde h_u$.
\\
Additional equivalence transformations between cases of the table are exhausted by the following:
\\
\ref{RDE:f=u_x,g=u^2}b${}\to{}$\ref{RDE:f=u_x,g=u^2}a: \
$\tilde t=\arctan t$, \
$\tilde x=x$, \
$\tilde u=(t^2+1)u+t$; \
\quad
\ref{RDE:f=u_x,g=u^2}c${}\to{}$\ref{RDE:f=u_x,g=u^2}a: \
$\tilde t=\frac12\ln|\frac{t-1}{t+1}|$, \
$\tilde x=x$, \
$\tilde u=( t^2-1)u+t$; \
\\
\ref{RDE:f=u_x^k,g=u^m}b${}\to{}$\ref{RDE:f=u_x^k,g=u^m}a$_{m=n+1}$, \ \ref{RDE:f=u_x^k,g=0}b${}\to{}$\ref{RDE:f=u_x^k,g=0}a: \
$\tilde t=e^{\varepsilon nt}/(\varepsilon n)$, \ $\tilde x=x$, \ $\tilde u=e^{-\varepsilon t}u$;  
\\
\ref{RDE:f=1,g=0}b${}\to{}$\ref{RDE:f=1,g=0}a: \ $\tilde t=t$, \ $\tilde x=x$, \ $\tilde u=e^{-\varepsilon t}u$;\quad
\ref{RDE:f=1,g=0}c${}\to{}$\ref{RDE:f=1,g=0}a: \ $\tilde t=t$, \ $\tilde x=x$, \ $\tilde u=u-t$;\quad
\ref{RDE:f=1,g=0}d${}\to{}$\ref{RDE:f=1,g=0}a: \ $\tilde t=t$, \ $\tilde x=u$, \ $\tilde u=x$.
\end{table}
\end{theorem}

Here $\mathcal G^\sim_{\mathcal L}$ denotes the equivalence groupoid of the class~$\mathcal L$.

\medskip

\noindent{\bf Notation for Table~\ref{tab:RDEGroupClassification}.}
Numbers with the same Arabic numerals and different Roman letters
correspond to cases that are equivalent with respect to additional equivalence transformations.
Explicit formulas for these transformations are presented in the footnote of Table~\ref{tab:RDEGroupClassification}.
The cases that are numbered with different numbers in Arabic numerals are reciprocally inequivalent with respect to point transformations.
Lie invariance algebras presented in Cases~\ref{RDE:f=any,g=any}, \ref{RDE:f=any,g=u} and~\ref{RDE:f=any,g=1/u}
are maximal only if the corresponding tuples of the arbitrary elements~$(f,g)$ are $G^\sim_{\mathcal R}$-inequivalent
to those from other cases.

The (usual) equivalence groups of the subclasses~$\mathcal H$ and~$\mathcal C$
coincide with each other and with the group~$G^\sim_{\mathcal R}$.
This is why within the class~$\mathcal R$,
Cases~\ref{RDE:f=u_x^k,g=e^u}$'$ and~\ref{RDE:f=u_x^k,g=u^m}a$'$
can be attached to Cases~\ref{RDE:f=u_x^k,g=e^u} and~\ref{RDE:f=u_x^k,g=u^m}a
by allowing the value $n=0$ in the latter cases.

Crossing out the numbers of Cases~\ref{RDE:f=1,g=0}c and~\ref{RDE:f=1,g=0}d indicates
that they are implicitly presented in the course of listing Lie-symmetry extensions for the subclass~$\mathcal F$
since both of them are $\bar G^\sim_{\mathcal F}$-equivalent to Case~\ref{RDE:f=1,g=0}a.

\begin{remark}
We should carefully treat arguments of logarithm and bases of powers with noninteger exponents.
The condition of positivity is justified for~$u$ by physical arguments
but this is not the case for~$u_x$.
Moreover, assuming $u$ positive makes impossible the change of the sign of~$u$,
which results in the disappearance of some discrete equivalence transformations.
This is why we use the absolute values of~$u$, $u_x$ and similar expressions
as arguments of logarithm and as bases of related powers with general real exponents.
This is not necessary for some particular exponents (more specifically, integer exponents and rational exponents with odd denominators)
but replacing an absolute value by the value itself in a power base
may lead to modifying the gauging of the coefficient of this power
since such a replacement may affect the action of discrete equivalence transformations on power coefficients.
\end{remark}

\begin{remark}
To construct an exhaustive list of $G^\sim_{\mathcal R}$-inequivalent Lie-symmetry extensions for the subclass~$\mathcal F$,
we should extend Cases~\ref{RDE:f=any,g=0}--
\ref{RDE:f=1,g=0} using transformations from the generalized equivalence group~$\bar G^\sim_{\mathcal F}$ of~$\mathcal F$
(or, equivalently, from the effective generalized equivalence group~$\hat G^\sim_{\mathcal F}$, which is more convenient)
and then gauge arising parameters by elements of~$G^\sim_{\mathcal R}$.
In each of the $G^\sim_{\mathcal R}$-inequivalent cases constructed, we can set $g\in\{0,1\}\bmod G^\sim_{\mathcal R}$,
and a complete list of $G^\sim_{\mathcal R}$-inequivalent values of the arbitrary element~$f$ is exhausted by
the general value (the extended Case~\ref{RDE:f=any,g=0}) and
\[
|u_x+\delta|^n,\quad
\frac{|u_x+\mu|^n}{|u_x+\nu|^{n+2}};\quad
e^{u_x},\quad
\frac{e^{(u_x+\lambda)^{-1}}}{(u_x+\lambda)^2};\quad
\frac{e^{m\arctan(u_x+\lambda)}}{(u_x+\lambda)^2+1};\quad
1,\quad
(u_x+\delta)^{-2}.
\]
Here $n\in\mathbb R\setminus\{0,-2\}$, $\delta\in\{0,1\}\bmod G^\sim_{\mathcal R}$,
$\lambda,\mu,\nu,m\in\mathbb R$,
$\mu\ne\nu$, and one nonzero constant among $\mu$ and $\nu$ can be set to be equal~$1$ modulo $G^\sim_{\mathcal R}$-equivalence,
$m\geqslant0\bmod G^\sim_{\mathcal R}$.
Semicolons in the above displayed equation separate the extensions of
Cases~\ref{RDE:f=u_x^k,g=0}, \ref{RDE:f=e^u_x,g=0}, \ref{RDE:fWithArctan,g=0} and \ref{RDE:f=1,g=0}, respectively.
In the course of implementing the above procedure for Case~\ref{RDE:fWithArctan,g=0},
we should take into account the formula for sum of arctangents,
\begin{gather*}
\arctan y+\arctan z=\left\{\begin{array}{ll}
\arctan\dfrac{y+z}{1-yz}+\dfrac\pi2(\sgn y)(1+\sgn(1-yz)),\quad&  \mbox{if}\quad yz\ne1,\\
\dfrac\pi2\sgn y,&  \mbox{if}\quad yz=1.
\end{array}
\right.
\end{gather*}
\end{remark}

\begin{remark}
Recall that the class~$\mathcal L$ is similar to the class~$\mathcal K$ of Kolmogorov equations
of the general form~\eqref{eq:RDEsKolmogorovEqs} with respect to a point transformation,
which is the hodograph transformation $\tilde t=t$, $\tilde x=u$, $\tilde u=x$.
This is why the groups classifications of the class~$\mathcal L$
up to $\mathcal G^\sim_{\mathcal L}$- and $G^\sim_{\mathcal L}$-equivalences reduce to
their counterparts for the class~$\mathcal K$;
see the theoretical background on mappings between classes of differential equations
that are generated by point transformations in~\cite{VaneevaPopovychSophocleous2009}.
For the part of Table~\ref{tab:RDEGroupClassification} related to the class~$\mathcal L$,
we use a complete list of inequivalent (up to general point equivalence) Lie-symmetry extensions
that is constructed for the general class of (1+1)-dimensional Kolmogorov equations
in~\cite[Corollary~7]{PopovychKunzingerIvanova2008}
and coincides with such a list for the class~$\mathcal K$.
The reason for this is that a similar list for the group classification of the class~$\mathcal K$
up to the equivalence generated by its equivalence group
is too cumbersome~\cite{OpanasenkoPopovych2018}.
\end{remark}

\begin{theorem}\label{thm:RDEGroupClassificationUpToGenPointEquiv}
A complete list of $\mathcal G^\sim_{\mathcal R}$-inequivalent (i.e., inequivalent up point transformations)
cases of Lie-symmetry extensions in the class~$\mathcal R$
is exhausted by Cases
\ref{RDE:f=any,g=u}, \ref{RDE:f=any,g=1/u},
\ref{RDE:f=u_x^k,g=e^u}${}\cup{}$\ref{RDE:f=u_x^k,g=e^u}$'$, \ref{RDE:f=u_x^k,g=u^m}a${}\cup{}$\ref{RDE:f=u_x^k,g=u^m}a$'$,
\ref{RDE:f=1/(u_x+1),g=u}, \ref{RDE:f=u_x,g=u^2}a, \ref{RDEPeculiarCase},
\ref{RDE:f=any,g=0}, \ref{RDE:f=u_x^k,g=0}a, \ref{RDE:f=e^u_x,g=0}, \ref{RDE:fWithArctan,g=0}, \ref{RDE:f=1,g=0}a,
and \ref{RDE:f=1,g=ulnu}--%
\noprint{\ref{RDE:f=1/u_x^2,g=any}, \ref{RDE:f=1/u_x^2,g=1/u}, }\ref{RDE:f=1/u_x^2,g=WithTan}
of Table~\ref{tab:RDEGroupClassification}.
\end{theorem}

\begin{proof}
To prove the $\mathcal G^\sim_{\mathcal R}$-inequivalence of the cases listed in this theorem,
we first use the fact that the maximal Lie invariance algebras of similar equations are similar realizations of isomorphic Lie algebras.
In particular, such algebras are of the same dimension.
Since the $t$-component of any admissible transformation in the class~$\mathcal R$ depends only on~$t$,
the dimension of the projection of the maximal Lie invariance algebra~$\mathfrak g$
of an equation from the class~$\mathcal R$ to the space with coordinate~$t$,
$\dim\pr_t\mathfrak g$, is one more invariant of~$\mathcal G^\sim_{\mathcal R}$.
We identify the structure of the associated maximal Lie invariance algebras, when they are finite dimensional,
according to Mubarakzyanov's classification of real low-dimensional Lie algebras~\cite{Mubarakzyanov1963a,Mubarakzyanov1963b}.
We use the notation of these algebras from~\cite{BoykoPateraPopovych2006,PopovychBoykoNesterenkoLutfullin2005},
additionally reparameterizing families of algebras by homogeneous parameters if this is convenient.

\medskip\par\noindent
$\dim\mathfrak g=2$.
Case~\ref{RDE:f=any,g=any}:     $\mathfrak g\simeq2A_1$;

\medskip\par\noindent
$\dim\mathfrak g=3$.
Case~\ref{RDE:f=any,g=u}:       $\mathfrak g\simeq A_{2.1}\oplus A_1$, $\dim\pr_t\mathfrak g=1$; \
Case~\ref{RDE:f=any,g=1/u}:     $\mathfrak g\simeq A_{3.4}^{[2:1]}$; \
Case~\ref{RDE:f=u_x^k,g=e^u}:   $\mathfrak g\simeq A_{3.3}$ if $n=-1$, and $\mathfrak g\simeq A_{3.4}^{[(n+2):1]}$ otherwise; \
Case~\ref{RDE:f=u_x^k,g=u^m}:   $\mathfrak g\simeq A_{2.1}\oplus A_1$, $\dim\pr_t\mathfrak g=2$ if $m=n+1$, $\mathfrak g\simeq A_{3.3}$ if $(1-m)(n+2)=n+1-m$, and $\mathfrak g\simeq A_{3.4}^{[(1-m)(n+2):(n+1-m)]}$ otherwise;

\medskip\par\noindent
$\dim\mathfrak g=4$.
Case~\ref{RDE:f=1/(u_x+1),g=u}: $\mathfrak g\simeq A_{4.8}^0$, $\dim\pr_t\mathfrak g=2$; \
Case~\ref{RDE:f=u_x,g=u^2}:     $\mathfrak g\simeq {\rm sl}_2(\mathbb R)\oplus A_1$; \
Case~\ref{RDEPeculiarCase}:     $\mathfrak g\simeq A_{3.4}^{[-1:1]}\oplus A_1$; \
Case~\ref{RDE:f=any,g=0}:       $\mathfrak g\simeq A_{4.5}^{[2:1:1]}$; \
Case~\ref{RDE:f=1,g=ulnu}:      $\mathfrak g\simeq A_{4.8}^0$, $\dim\pr_t\mathfrak g=1$;

\medskip\par\noindent
$\dim\mathfrak g=5$.
Case~\ref{RDE:f=u_x^k,g=0}:     $\mathfrak g\simeq A_{5.33}^{n+2,-n}$;
Case~\ref{RDE:f=e^u_x,g=0}:     $\mathfrak g\simeq A_{5.34}^2$;
Case~\ref{RDE:fWithArctan,g=0}: $\mathfrak g\simeq A_{5.35}^{2,m}$.

\medskip\par\noindent
$\dim\mathfrak g=\infty$.
Cases~\ref{RDE:f=1,g=0}, \ref{RDE:f=1/u_x^2,g=any}, \ref{RDE:f=1/u_x^2,g=1/u} and  \ref{RDE:f=1/u_x^2,g=WithTan}
are then inequivalent to other cases.
The inequivalence of these cases to each other and the impossibility of further gauging of the parameters~$\mu$ and~$\nu$
follows from the same properties of similar Kolmogorov equations \cite[Corollary~7]{PopovychKunzingerIvanova2008}.

\medskip

We discuss only unobvious $\mathcal G^\sim_{\mathcal R}$-inequivalences.

Admissible transformations of the class~$\mathcal R$ preserve the $t$-direction.
This is why the order of parameters is essential in each of the above algebras from the family~$A_{3.4}$.
Since $n\notin\{0,-2\}$ and, in Case~\ref{RDE:f=u_x^k,g=u^m}, $m\ne-1$,
Cases \ref{RDE:f=u_x^k,g=e^u} and~\ref{RDE:f=u_x^k,g=u^m} cannot be mapped by elements of~$\mathcal G^\sim_{\mathcal R}$
to Case~\ref{RDE:f=any,g=1/u}.

Algebras $A_{5.35}^{2,m}$ and \smash{$A_{5.35}^{2,m'}$} with different to each other nonnegative~$m$ and~$m'$ are non-isomorphic.
Hence the equations of Case~\ref{RDE:fWithArctan,g=0} with such~$m$'s are $\mathcal G^\sim_{\mathcal R}$-inequivalent.

It is still necessary to prove that equations
from different Cases~\ref{RDE:f=u_x^k,g=e^u}${}\cup{}$\ref{RDE:f=u_x^k,g=e^u}$'$
and~\ref{RDE:f=u_x^k,g=u^m}a${}\cup{}$\ref{RDE:f=u_x^k,g=u^m}a$'$
are $\mathcal G^\sim_{\mathcal R}$-inequivalent to each other,
and different values of~$n$ in Case~\ref{RDE:f=u_x^k,g=e^u}${}\cup{}$\ref{RDE:f=u_x^k,g=e^u}$'$
and different values of~$(n,m)$ in Case~\ref{RDE:f=u_x^k,g=u^m}a${}\cup{}$\ref{RDE:f=u_x^k,g=u^m}a$'$
are $\mathcal G^\sim_{\mathcal R}$-inequivalent.
In Cases \ref{RDE:f=u_x^k,g=e^u}${}\cup{}$\ref{RDE:f=u_x^k,g=e^u}$'$,
\ref{RDE:f=1/(u_x+1),g=u}, \ref{RDEPeculiarCase} and \ref{RDE:f=1,g=ulnu},
the values $\varepsilon=1$ and $\varepsilon=-1$
should be $\mathcal G^\sim_{\mathcal R}$-inequivalent.
We also should check that in Case~\ref{RDE:f=u_x^k,g=0}a,
different values of~$n$ in $[-1,+\infty)\setminus\{0,2\}$ are $\mathcal G^\sim_{\mathcal R}$-inequivalent.
Further we use the direct method.

Consider a point transformation of independent and dependent variables of the general form~\eqref{eq:RDEsGenPointTrans}
that connects the source and the target equations from the class~$\mathcal R$, \eqref{eq:RDEsSource&TargetEqs}
with $f(u_x)=|u_x|^n$ and $\tilde f(\tilde u_{\tilde x})=|\tilde u_{\tilde x}|^{\tilde n}$,
where $n,\tilde n\in\mathbb R\setminus\{0,-2\}$.
We substitute these values of the arbitrary element~$f$ into the equation~\eqref{eq:RDEsExprForSourceArbitraryElementF}
and act on the obtained equation by the operator $\p_{u_x}-nu_x^{-1}$.
Rearranging the result, we derive the equation
\[
\tilde n\frac{\Delta}{D_xU}-n\frac{D_xX}{u_x}-2X_u=0,
\]
whose left hand side is rational in~$u_x$.
There are two cases depending on the value of~$U_u$:
\begin{enumerate}\itemsep=0ex
\item
$U_u=0$. Then $X_uU_x\ne0$, $X_x=0$ and $\tilde n=-n-2$.
\item
$U_u\ne0$. Then $U_x=0$, $X_x\ne0$, $X_u=0$ and $\tilde n=n$.
\end{enumerate}

In the first case, the equation~\eqref{eq:RDEsExprForSourceArbitraryElementF}
reduces to $T_t|X_u|^n/|U_x|^{n+2}=1$, which further implies that $X_{uu}=U_{xx}=0$.
Then $V_x=V_u=0$, and the equation~\eqref{eq:RDEsExprForSourceArbitraryElementG}
splits with respect to~$u_x$ into $g=-X_t/X_u$ and $\tilde g=U_t/T_t$,
i.e., $g$ and~$\tilde g$ are affine in~$u$.
Therefore, the source and the target equations fit
into Cases~\ref{RDE:f=u_x^k,g=0}a or~\ref{RDE:f=u_x^k,g=0}b up to \smash{$G^\sim_{\mathcal R}$-equivalence} and
merely into Case~\ref{RDE:f=u_x^k,g=0}a up to \smash{$\mathcal G^\sim_{\mathcal R}$-equivalence}.
This completes the proof for Case~\ref{RDE:f=u_x^k,g=0}.
Notes that algebras $A_{5.33}^{a,b}$ and \smash{$A_{5.33}^{a',b'}$} are isomorphic if and only if
\[
(a',b')\in\{(a,b),(b,a)\}
\cup\bigg\{\bigg(\frac1a,-\frac ba\bigg),\bigg(-\frac ba,\frac1a\bigg)\bigg\}_{{\rm if\,}a\ne0}\!\!
\cup\bigg\{\bigg(-\frac ab,\frac 1b\bigg),\bigg(\frac 1b,-\frac ab\bigg)\bigg\}_{{\rm if\,}b\ne0}.
\]
Hence inequivalence of equations within Case~\ref{RDE:f=u_x^k,g=0} cannot be checked algebraically
with comparing structure of the corresponding Lie invariance algebras.

The second case is analyzed in a similar way.
The equation~\eqref{eq:RDEsExprForSourceArbitraryElementF} reduces to the equation $T_t|U_u|^n/|X_x|^{n+2}=1$
implying $X_{xx}=U_{uu}=0$ and $T_t>0$.
Then again $V_x=V_u=0$, and the equation~\eqref{eq:RDEsExprForSourceArbitraryElementG}
splits with respect to~$u_x$ into $X_t=0$ and $g=(T_t\tilde g-U_t)/U_u$.
It follows from the last equation that point transformations cannot switch
an exponential value of~$g$ to a power one and conversely
as well as they cannot change the value of the exponent~$m$
when the arbitrary element~$g$ is of power form $g=|u|^m$.
Moreover, if $g=\varepsilon e^u$ and $\tilde g=\tilde\varepsilon e^{\tilde u}$,
then $U_u=1$ and one cannot alternate the sign of~$\varepsilon$.
In other words, there are no more additional equivalence transformations
related to equations of Cases~\ref{RDE:f=u_x^k,g=e^u} and~\ref{RDE:f=u_x^k,g=u^m}a.

Up to extending the argument tuples of~$g$ and~$\tilde g$,
the equations~\eqref{eq:RDEsExprForSourceArbitraryElementF} and~\eqref{eq:RDEsExprForSourceArbitraryElementG}
preserve their forms if we consider a wider class of differential equations,
allowing for the arbitrary element~$g$ to additionally depend on~$t$, $x$ and~$u_x$.
This is why for simplifying the analysis of Cases~\ref{RDE:f=1/(u_x+1),g=u} and~\ref{RDEPeculiarCase},
we map the associated equations, $u_t=(u_x+1)^{-1}u_{xx}+\varepsilon u$ and $u_t=u_x(u_x+1)^{-3}u_{xx}+\varepsilon u$,
to simpler equations, $u_t=u_x^{-1}u_{xx}+\varepsilon(u-x)$ and $u_t=u_xu_{xx}-\varepsilon uu_x$,
by the point transformations
$\tilde t=t$, $\tilde x=x$, $\tilde u=u+x$ and
$\tilde t=\e^{-1}e^{\e t}$, $\tilde x=x+u$, $\tilde u=e^{-\e t}u$,
respectively, where tildes indicate variables of target equations.
(The same mapping of Case~\ref{RDEPeculiarCase} is used in Section~\ref{sec:RDEsSolutions}, cf.\ the equation~\eqref{eq:RDEsPeculiarQuadratic}.)
Looking for a point transformation between target equations of the same form with different values of~$\varepsilon$,
well fits into the above second case.
For the first target equation, we have $n=\tilde n=-1$,  $g=\varepsilon(u-x)$ and $\tilde g=\tilde\varepsilon(\tilde u-\tilde x)$,
and collecting the coefficients of~$u$ in the equation~\eqref{eq:RDEsExprForSourceArbitraryElementG}
leads to the equation $\tilde\varepsilon=T_t\varepsilon$
implying the impossibility of alternating between the values $\varepsilon=1$ and~$\varepsilon=-1$ since $T_t>0$.
For the second target equation, we similarly have $n=\tilde n=1$,  $g=-\varepsilon uu_x$ and $\tilde g=-\tilde\varepsilon\tilde u\tilde u_{\tilde x}$,
and collecting the coefficients of~$uu_x$ in the equation~\eqref{eq:RDEsExprForSourceArbitraryElementG}
gives the equation $\tilde\varepsilon=T_tU_uX_x^{-1}\varepsilon$.
Jointly with the consequence $T_tU_u=X_x^3$ of the equation~\eqref{eq:RDEsExprForSourceArbitraryElementF},
this again implies the impossibility of alternating the sign of~$\varepsilon$.

Let now $f=1$ and $\tilde f=1$.
Then the equation~\eqref{eq:RDEsExprForSourceArbitraryElementF} splits
with respect to~$u_x$ into the equations $X_u=0$ and $T_t=X_x^2$,
which imply $X_xU_u\ne0$, $T_t>0$ and $X_{xx}=0$.
Collecting coefficients of different powers of~$u_x$ in the equation~\eqref{eq:RDEsExprForSourceArbitraryElementG}
results in the equations
$U_{uu}=0$, $2X_xU_{xu}+X_tU_u=0$ and $U_ug=T_t\tilde g+U_{xx}-U_t+X_tU_x/X_x$.
We differentiate the last equation twice with respect to~$u$
in order to derive its differential consequence $g_{uu}=T_tU_u\tilde g_{\tilde u\tilde u}$.
Therefore, elements of~$\mathcal G^\sim_{\mathcal R}$ preserve each of
Cases~\ref{RDE:f=u_x^k,g=e^u}$'$, \ref{RDE:f=u_x^k,g=u^m}a$'$ and \ref{RDE:f=1,g=ulnu}
modulo $G^\sim_{\mathcal R}$-equivalence including the value of~$m$.
Moreover, if $g=\varepsilon e^u$ and $\tilde g=\tilde\varepsilon e^{\tilde u}$,
then $U_u=1$ and $\tilde\varepsilon=T_t\varepsilon e^{u-U}$,
i.e., one cannot alternate the sign of~$\varepsilon$.
Analogously, for $g=\varepsilon u\ln|u|$ and $\tilde g=\tilde\varepsilon\tilde u\ln|\tilde u|$
we have $\tilde\varepsilon=T_t\varepsilon$.
This means that the set of
Cases~\ref{RDE:f=u_x^k,g=e^u}$'$, \ref{RDE:f=u_x^k,g=u^m}a$'$ and \ref{RDE:f=1,g=ulnu}
also admits no additional equivalence transformations.
\end{proof}

\section{Method of furcate splitting}\label{sec:RDEsFurcateSplit}

Although the literature dedicated to group classification of classes (of systems) of differential equations is vast,
a considerable part thereof is not trustworthy at all.
The main problem lies in the fact that often a brute-force approach is applied to tackle group classification problems.
Albeit this approach is definitely applicable (but certainly inefficient) for some simple classes of differential equations,
it requires a very thorough inspection to keep track of various cases that arise
upon integrating the corresponding determining equations for Lie symmetries,
otherwise it leads to missed, overlapping, repeated and/or incorrect classification cases.
These flaws happened in the majority of papers where the brute-force approach is used for group classification of differential equations.
Moreover, since this approach is often cumbersome, such papers do not contain details of solutions,
which makes pretty hard to find flaws in the computations afterwards.

While the class~$\mathcal C$ does not possess a necessary structure to employ the powerful algebraic method of group classification,
a brute force is still not the last resort.
The best choice here is the method of \emph{furcate splitting}.
This method was suggested and applied for the first time in~\cite{NikitinPopovych2001}
but its name appeared later in~\cite{IvanovaPopovychSophocleous2010}.
See also examples of its application in~\cite{PopovychIvanova2004,VaneevaKurikshaSophocleous2015,VaneevaPopovychSophocleous2012,VaneevaSophocleousLeach2015}.
It is basically a refinement of the direct approach to group classification.
In contrast to the chaotical integration of the classifying equations
and the irregular selection of cases with Lie-symmetry extensions within the framework of the direct approach,
the method of furcate splitting provides an algorithm for systematically finding such cases.
It is worth to note that the method is effective for classes with arbitrary elements depending on a few arguments.
So far, it was used for classes with arbitrary elements depending on at most two arguments~\cite{NikitinPopovych2001}.
This limit for manually using the method seems to be reasonable
in view of the necessity of verifying the compatibility condition of arising equations,
which becomes the more difficult problem the more arguments the arbitrary elements depend on.

The method of furcate splitting can be easily understood by examples
but the description of its standard version in the general setting is quite tedious,
not to mention its various modifications.
Nevertheless, below we present this description,
and in Section~\ref{sec:RDEsGroupClassificationOfC}
we illustrate possible complications for using the method
if arbitrary elements of the class under study depend on different arguments,
which is the case for the class~$\mathcal R$;
see also the related discussion in the end of Section~\ref{sec:RDEsConclusions}.
\looseness=-1

Consider a class of (systems of) differential equations
\[
\mathscr L|_{\mathscr S}=\{\mathscr L_\theta\colon L(x,u_{(r)},\theta(x,u_{(r)}))=0\mid \theta\in\mathscr S\}.
\]
Here and below in this section the notation differs from the other sections of the paper:
$x$ and $u$ are the tuples of independent and dependent variables, respectively,
constituting the coordinates for a coordinate space,
$u_{(r)}$ denotes the tuple of derivatives of~$u$ with respect to $x$ up to order~$r$,
which also includes the~$u$'s as the derivatives of order zero,
and $L$ is a tuple of differential functions in~$u$,
parameterized by the tuple of functions~$\theta=(\theta^1(x,u_{(r)}),\dots,\theta^p(x,u_{(r)}))$,
called the arbitrary elements of the class~$\mathscr L|_{\mathscr S}$.
The arbitrary-element tuple runs through the solution set~$\mathscr S$ of an auxiliary system~${\rm AS}$
of differential equations and inequalities in~$\theta$,
$S(x,u_{(r)},\theta_{(q)}(x,u_{(r)}))=0$ and, e.g., $\Sigma(x,u_{(r)},\theta_{(q)}(x,u_{(r)}))\ne0$,
where $(x,u_{(r)})$ is assumed to be the tuple of independent variables,
and the notation $\theta_{(q)}$ encompasses the partial derivatives of the arbitrary elements~$\theta$ up to order $q$ with respect to both~$x$ and~$u_{(r)}$.
See \cite{Popovych2006,PopovychKunzingerEshraghi2010} for more detailed explanations
on the notion of class of differential equations.

Denote by $\mathfrak g_\theta$ the maximal Lie invariance algebra of a system~$\mathscr L_\theta\in\mathscr L|_{\mathscr S}$.
The system of determining equations ${\rm DE}[\theta]$ for components of vector fields in~$\mathfrak g_\theta$
may split into two subsystems.
The equations of one of these subsystems, ${\rm CE}[\theta]$, essentially involve arbitrary elements and are called \emph{classifying equations},
and the other subsystem, ${\rm SE}$, does not involve~$\theta$ and is thus shared by all $\mathscr L_\theta\in\mathscr L|_{\mathscr S}$.
The subsystem~${\rm SE}$ can be integrated directly,
thus specifying the general form of Lie-symmetry vector fields of systems in~$\mathscr L|_{\mathscr S}$.
The subsystem ${\rm CE}[\theta]$ with a fixed value of the arbitrary-element tuple~$\theta$
is the system of specific equations for the components of vector fields in~$\mathfrak g_\theta$.
The solution set of the joint system~${\rm DE}[\theta]$ is a linear space in view of the fact
that ${\rm DE}[\theta]$ is a homogeneous linear system of differential equations
with respect to the components of Lie-symmetry vector fields.
Conversely, fixing a vector field~$Q$ on the space with coordinate $(x,u)$,
we obtain a system of differential equations for the values of~$\theta$
for which the corresponding system~$\mathscr L_\theta$ admits $Q$ as a Lie-symmetry vector field.

The system~${\rm AS}$ as a rule includes equations implying the independence of~$\theta$
on certain independent or dependent variables or, more generally, on certain functional combinations of these variables.
It is convenient to replace the coordinates $(x,u_{(r)})$ in the $r$th order jet space by the (local-)coordinate tuple~$z$
that is split into two subtuples, $\hat z$ and~$\check z$, such that
$\hat z$ is a maximal tuple of coordinates not involved in~$\theta$,
and each of the tuples~$\hat z$ and~$\check z$ is split into two subtuples,
$\hat z'$, $\hat z''$ and $\check z'$, $\check z''$,
where $\hat z'$ and $\check z'$ are maximal subtuples of~$\hat z$ and~$\check z$, respectively,
that appear in arguments of functions parameterizing the general solution of~${\rm SE}$.
We choose the tuple $\breve z'$ of functions of~$(x,u)$ such that
the joint tuple $(\hat z',\check z',\breve z')$ gives new coordinates in the space of~$(x,u)$.
In most of practical computations by the method of furcate splitting,
the optimal choice of $\hat z'$, $\hat z''$, $\check z'$, $\check z''$ and~$\breve z'$ is obvious,
cf.\ Section~\ref{sec:RDEsGroupClassificationOfC}.

We substitute the expressions for the components of Lie-symmetry vector fields that are specified by ${\rm SE}$
into ${\rm CE}[\theta]$ and split the obtained system with respect to $\hat z''$.
This leads to the system ${\rm CE}'[\theta]$ for the constants and functions parameterizing these expressions, 
which is supposed to be of the form 
\[
\sum_{i=i_{s-1}+1}^{i_s}\psi^iF^i=0, \quad s=1,\dots,|{\rm CE}'|,
\]
where $0=i_0<i_1<\dots<i_{|{\rm CE}'|}=N$ with $N\in\mathbb N$,
and $|{\rm CE}'|$ denotes the number of (independent) equations in~${\rm CE}'[\theta]$.
The functions $F^i$, $i=1,\dots,N$, are known differential functions of~$\theta$
with $\check z$ assumed to be the tuple of independent variables.
The coefficients $\psi^i$, $i=1,\dots,N$, may depend on $(\hat z',\check z',\breve z')$
and are in fact parameterized by an element of~$\mathfrak g_\theta$.

Let us imagine that the algebra~$\mathfrak g_\theta$ is known for each fixed~$\theta$.
Substituting values of the above parameters corresponding to a vector field in~$\mathfrak g_\theta$
and of the variables~$\hat z'$ into the system~${\rm CE}'[\theta]$
and varying the vector field within~$\mathfrak g_\theta$ and the variables~$\hat z'$
give a family (denoted by~${\rm TFF}$) of systems for~$\theta$ of the same general form called the \emph{template form},
\[
\sum_{i=i_{s-1}+1}^{i_s}a^iF^i=0, \quad s=1,\dots,|{\rm CE}'|,
\]
where coefficients $a^i$, $i=1,\dots,N$, may depend on $(\check z',\breve z')$.
In general, these coefficients are not precisely known at this stage but they may satisfy known constraints.

Let $k$ be the maximal number of systems in~${\rm TFF}$ such that
the rank of the collection of the coefficient tuples $\bar a^q=(a^{q1},\dots,a^{qN})$ of these systems
equals their number, $\mathop{\rm rank}A=k$ with $A:=(a^{qi})_{q=1,\dots,k}^{i=1,\dots,N}$.
It is obvious that $0\leqslant k\leqslant N$.
The condition of consistency of systems within~${\rm TFF}$ and of these systems with~${\rm AS}$
additionally majorizes the possible values of~$k$.

Supposing a certain value for~$k$ within the range of possible values of~$k$,
we study the consistency of~${\rm TFF}$.
The consideration can be partitioned into cases by choosing a sequence of
$k\times k$ minors of the matrix~$A$
in a way consistent with the relevant equivalence within the class~$\mathscr L|_{\mathscr S}$
and by successively supposing
that a minor in the sequence does not vanish and all the preceding minors are zero.
Some of the coefficients~$a^{qi}$ can be gauged using equivalence transformations
in view of the made supposition or derived constraints at this or further steps.
For each of the cases, one should
\begin{itemize}\itemsep=0ex
\item
derive the conditions implied by the consistency of~${\rm TFF}$ on the coefficients~$a^{qi}$
and additionally gauge these coefficients by equivalence transformations (if possible),
\item
split the system~${\rm CE}'[\theta]$ on the solution set of the joint system of~${\rm AS}$ and~${\rm TFF}$
with respect to parametric derivatives of~$\theta$,
\item
solve the obtained system of additional determining equations for still unspecified parameters in
Lie symmetry vector fields,
\item
derive the precise form of~${\rm TFF}$ for the computed algebra of vector fields
and check its consistency with the supposed value of~$k$ and other suppositions,
\item
if the consistency holds, solve the specified system~${\rm TFF}$ with respect to~$\theta$.
\end{itemize}

One has $k=0$ if and only if the classifying equations are identically satisfied in view of~${\rm SE}$,
which corresponds to the general case without Lie-symmetry extension.

The value $k=1$ is associated with minimal Lie-symmetry extensions.
For this value of~$k$, the second step of the above procedure is convenient to be carried out in a special way
since it means that the tuple $\bar\psi=(\psi^1,\dots,\psi^N)$ is proportional to $\bar a^1$, $\psi^i=\lambda a^{1i}$.
Here the multiplier~$\lambda$ may depend on~$(\hat z',\check z',\breve z')$ or, equivalently, on $(x,u)$,
and this dependence can be easily specified in view of the form of~$\bar\psi$ and~$\bar a^1$.

\section{Solution of group classification problem for regular subclass}\label{sec:RDEsGroupClassificationOfC}

In this section we prove the following assertion.

\begin{proposition}\label{pro:RDEsGroupClassificationOfC}
A complete list of $G^\sim_{\mathcal R}$-inequivalent Lie-symmetry extensions
in the class~$\mathcal C$ is exhausted by Cases~
\ref{RDE:f=any,g=u},
\ref{RDE:f=any,g=1/u},
\ref{RDE:f=u_x^k,g=e^u},
\ref{RDE:f=u_x^k,g=u^m}a,
\ref{RDE:f=u_x^k,g=u^m}b,
\ref{RDE:f=1/(u_x+1),g=u},
\ref{RDE:f=u_x,g=u^2}a,
\ref{RDE:f=u_x,g=u^2}b,
\ref{RDE:f=u_x,g=u^2}c,
\ref{RDEPeculiarCase} and
\ref{RDE:f=u_x^k,g=0}b
of Table~\ref{tab:RDEGroupClassification}.
\end{proposition}

Recall that the $G^\sim_{\mathcal C}$-equivalence coincides
with the restriction of the $G^\sim_{\mathcal R}$-inequivalence to the class~$\mathcal C$.

It turns out that additionally to $\tau_x=\tau_u=0$,
components of Lie-symmetry vector fields of equations in the class~$\mathcal C$
satisfy more determining equations not involving the arbitrary elements~$f$ and~$g$.

\begin{lemma}\label{lem:RDEsDetEqsForC}
Under the conditions $f_{u_x}\ne0$, $(u_x{}^{\!2}f)_{u_x}\ne0$ and $g_u\ne0$, the system~\eqref{eq:RDEsDetEqSystem} implies 
\begin{gather}\label{eq:RDEsSpecificDetEqsForClassC1}
\xi_{xx}=\xi_{xu}=\xi_{uu}=\eta_{xx}=\eta_{xu}=\eta_{uu}=0.
\end{gather}
\end{lemma}

\begin{proof}
We consider two obvious differential consequences of the system~\eqref{eq:RDEsDetEqSystem}.
One of the consequences is derived by acting the operators $u_x\p_u+\p_x$ and ${}-\p_{u_x}$
on the equations~\eqref{eq:RDEsDetEqA} and~\eqref{eq:RDEsDetEqB}, respectively,
and summing the obtained equations, which gives
\begin{gather}\label{eq:RDEsDetEqC}
\big(\xi_{uu}u_x{}^{\!2}-2\eta_{uu}u_x-2\eta_{xu}-\xi_{xx}\big)f=\xi_t+\xi_ug.
\end{gather}
The other consequence is the sum of the equation~\eqref{eq:RDEsDetEqB}
with the equation~\eqref{eq:RDEsDetEqC} multiplied by~$u_x$,
\begin{gather}\label{eq:RDEsDetEqD}
\big((\eta_{uu}+2\xi_{xu})u_x{}^{\!2}+2\xi_{xx}u_x-\eta_{xx}\big)f=\eta g_u+(\tau_t-\eta_u)g-\eta_t.
\end{gather}
The equations~\eqref{eq:RDEsDetEqC} and~\eqref{eq:RDEsDetEqD} hint that the proof is partitioned into three cases,
depending on whether these equations imply conditions on~$f$ and what structures of such conditions are,
\[
1.\ (1/f)_{u_xu_xu_x}\ne0,\quad 2.\ (1/f)_{u_xu_xu_x}=0,\ (1/f)_{u_xu_x}\ne0, \quad\mbox{and}\quad 3.\ (1/f)_{u_xu_x}=0.
\]

In the first case, we can split the above differential consequences with respect to both $f$ and~$u_x$,
which directly leads to the required equations.

The third condition means that $f=(u_x+\gamma)^{-1}\bmod G^\sim_{\mathcal R}$.
(Recall that we suppose the inequality $f_{u_x}\ne0$.)
Substituting the expression for~$f$ into~\eqref{eq:RDEsDetEqA} and splitting with respect to~$u_x$
give the equations $\xi_u=0$, $\eta_u=\gamma(\tau_t-\xi_x)$ and $\eta_x=\tau_t-\xi_x$,
which imply $\eta_{xu}=\eta_{uu}=0$.
Using derived equations, we simplify the equation~\eqref{eq:RDEsDetEqC} and then treat it in the same way,
which in particular gives $\xi_{xx}=0$, and thus also $\eta_{xx}=0$.

In the second case, which is much more complicated than the other cases,
modulo $G^\sim_{\mathcal R}$-equivalence we can set $f=(u_x{}^{\!2}+2\beta u_x+\gamma)^{-1}$,
where $(\beta,\gamma)\ne(0,0)$ since $(u_x{}^{\!2}f)_{u_x}\ne0$.
Substituting the expression for~$f$ into~\eqref{eq:RDEsDetEqA}, splitting with respect to~$u_x$
and arranging the obtained equations, we get
\begin{gather}\label{eq:RDEsDetEqsForFujitaNonlinearity}
\eta_u=-\beta\xi_u+\frac12\tau_t,\quad
\eta_x=-\beta\xi_x+\frac{\beta}2\tau_t+(\beta^2-\gamma)\xi_u,\quad
(\beta^2-\gamma)(2\xi_x-\beta\xi_u-\tau_t)=0.
\end{gather}
The cross differentiation of the two first equations in~\eqref{eq:RDEsDetEqsForFujitaNonlinearity}
implies $(\beta^2-\gamma)\xi_{uu}=0$.
Then the separate differentiations of the last equation in~\eqref{eq:RDEsDetEqsForFujitaNonlinearity}
with respect to~$x$ and~$u$ successively give $(\beta^2-\gamma)\xi_{xu}=0$ and $(\beta^2-\gamma)\xi_{xx}=0$.
In view of these equations for~$\xi$,
the same differentiation of the two first equations in~\eqref{eq:RDEsDetEqsForFujitaNonlinearity} results in
$\eta_{xx}=-\beta\xi_{xx}$, $\eta_{xu}=-\beta\xi_{xu}$ and $\eta_{uu}=-\beta\xi_{uu}$.
Therefore, if $\gamma\ne\beta^2$, then we have the required determining equations.

Suppose that $\gamma=\beta^2$.
Then $\beta\ne0$, and modulo $G^\sim_{\mathcal R}$-equivalence we can set $\beta=1$, i.e., $f=(u_x+1)^{-2}$.
The general solution of the system $\eta_x=-\xi_x+\tau_t/2$, $\eta_u=-\xi_u+\tau_t/2$
with respect to~$\eta$ is $\eta=-\xi+\tau_t(u+x)/2+\eta^0(t)$,
where $\eta^0$ is an arbitrary smooth function of~$t$.
For these values of~$f$ and~$\eta$, the determining equation~\eqref{eq:RDEsDetEqB} takes the form
\[
\big((\xi_t+\xi_ug)u_x+\eta g_u+(\tau_t-\eta_u)g-\eta_t\big)(u_x+1)=\xi_{uu}u_x{}^{\!2}+2\xi_{xu}u_x+\xi_{xx}
\]
and splits with respect to~$u_x$ into the system
\begin{gather}\label{eq:RDEsIntermediateDetEqsForF=(u_x+c)^-2}
\xi_t+\xi_ug=\xi_{uu},\quad
\eta g_u+(\tau_t-\eta_u)g-\eta_t=\xi_{xx},\quad
\xi_{uu}-2\xi_{xu}+\xi_{xx}=0.
\end{gather}
The last equation can be represented as $(\p_x-\p_u)^2\xi=0$,
and hence its general solution is $\xi=\xi^1(t,\omega)u+\xi^0(t,\omega)$,
where $\xi^1$ and~$\xi^0$ are arbitrary smooth functions of~$t$ and $\omega=x+u$.
In view of the derived expressions for~$\xi$ and~$\eta$, the system of the first two equations of~\eqref{eq:RDEsIntermediateDetEqsForF=(u_x+c)^-2}
reduces~to the system
\begin{gather}
\xi^1_tu+\xi^0_t+(\xi^1_\omega u+\xi^1+\xi^0_\omega)g=\xi^1_{\omega\omega}u+2\xi^1_\omega+\xi^0_{\omega\omega},\label{eq:RDEsTemp1}\\
(-2\xi^1u-2\xi^0+\tau_t\omega+2\eta^0)g_u+\tau_tg+4\xi^1_\omega-\tau_{tt}\omega-2\eta^0_t=0.\label{eq:RDEsTemp2}
\end{gather}
We assume $(t,\omega,u)$ to be the tuple of independent variables in the last system.

If the arbitrary element~$g$ is not a fractional linear function,
then we can split the equation~\eqref{eq:RDEsTemp1} simultaneously with respect to~$g$ an~$u$
and obtain the equations $\xi^1_\omega=0$ and~$\xi^0_\omega=-\xi^1$, whose consequence is $\xi^0_{\omega\omega}=0$.
Otherwise (i.e., in the case of nonconstant fractional linear~$g$, since $g_u\ne0$ by a lemma's assumption),
we differentiate the equation~\eqref{eq:RDEsTemp2} with respect to~$\omega$ and
split the derived differential consequence with respect to~$u$,
which again leads to the equations $\xi^1_\omega=0$ and~$\xi^0_{\omega\omega}=0$.
Therefore, in the third case we also have the required equations
$\xi_{xx}=\xi_{xu}=\xi_{uu}=0$ and, therefore, $\eta_{xx}=\eta_{xu}=\eta_{uu}=0$.
\end{proof}

In other words, the system of determining equations for Lie symmetries of equations in the class~$\mathcal C$
in fact reduces to the system of the equations $\tau_x=\tau_u=0$, \eqref{eq:RDEsDetEqA}, \eqref{eq:RDEsSpecificDetEqsForClassC1} and
\begin{gather}\label{eq:RDEsSpecificDetEqsForClassC2}
\xi_t+\xi_ug=0,
\\ \label{eq:RDEsSpecificDetEqsForClassC3}
\eta g_u+(\tau_t-\eta_u)g-\eta_t=0.
\end{gather}
The system~\eqref{eq:RDEsSpecificDetEqsForClassC1} is equivalent to the following representation for the components~$\xi$ and~$\eta$:
\begin{gather*}
\xi=\xi^1(t)x+\xi^2(t)u+\xi^0(t),\quad
\eta=\eta^1(t)x+\eta^2(t)u+\eta^0(t),
\end{gather*}
where $\xi^0$,  $\xi^1$,  $\xi^2$, $\eta^0$,  $\eta^1$,  and $\eta^2$ are smooth functions of~$t$.

\begin{lemma}\label{lem:RDEsSingularCasesInC}
(i) For all equations in the class~$\mathcal C$,
except those $G^\sim_{\mathcal R}$-equivalent to equations of Case~\ref{RDEPeculiarCase},
the $x$-components of admitted Lie-symmetry vector fields satisfy the determining equation $\xi_u=0$.
(ii) For all equations in the class~$\mathcal C$,
except those $G^\sim_{\mathcal R}$-equivalent to equations of Case~\ref{RDE:f=1/(u_x+1),g=u},
the $u$-components of admitted Lie-symmetry vector fields satisfy the determining equation $\eta_x=0$.
\end{lemma}

\begin{proof}
Suppose that an equation~$\mathfrak E$ in the class~$\mathcal C$ admits a Lie-symmetry vector field~$Q$ with $(\xi^2,\eta^1)\ne(0,0)$,
and $(f,g)$ is the associated value of the arbitrary-element tuple.

If $\xi_u=\xi^2\ne0$ (resp.\ $\eta_x=\eta^1\ne0$) for~$Q$, then the equation~\eqref{eq:RDEsSpecificDetEqsForClassC2}
(resp.\ the differential consequence of~\eqref{eq:RDEsSpecificDetEqsForClassC3} obtained with the action by the operator~$\p_x$)
implies $g_{uu}=0$, and hence $g=u\bmod G^\sim_{\mathcal R}$ since $g_u\ne0$ in the class~$\mathcal C$.
For this value of the arbitrary element~$g$,
the equations~\eqref{eq:RDEsSpecificDetEqsForClassC2} and~\eqref{eq:RDEsSpecificDetEqsForClassC3}
respectively split with respect to~$(x,u)$ into the equations
\begin{gather}\label{eq:RDEsDetEqsForclassCAndLinG}
\xi^0_t=\xi^1_t=0,\quad \xi^2_t=-\xi^2\quad\mbox{and}\quad
\eta^0_t=\eta^0,\quad \eta^1_t=\eta^1,\quad \eta^2_t=\tau_t.
\end{gather}

Now we apply the furcate splitting with respect to~$f$.
Here the template form of equations for~$f$ is
\begin{gather}\label{eq:RDEsTemplateEqForF}
(a_1u_x{}^{\!2}+a_2u_x+a_3)f_{u_x}+(2a_1u_x+a_4)f=0,
\end{gather}
where $a_1$, \dots, $a_4$ are constants.
Template-form equations for~$f$ are obtained by fixing~$t$ in the classifying equation~\eqref{eq:RDEsDetEqA}.
Denote by~$k$ the number of linearly independent tuples $(a_1,a_2,a_3,a_4)$
among those associated with template-form equations for~$f$.
We have $k>0$ since $(\xi^2,\eta^1)\ne(0,0)$ for the Lie-symmetry vector field~$Q$.

If $k\geqslant2$, then $f=\varepsilon(u_x+1)^{-2}\bmod  G^\sim_{\mathcal R}$ with $\varepsilon=\pm1$
since $f_{u_x}\ne0$ and $(u_x{}^{\!2}f)_{u_x}\ne0$ for equations in the class~$\mathcal C$.
Splitting the equation~\eqref{eq:RDEsDetEqA} for this value of~$f$ with respect to~$u_x$,
we derive the equations $\tau_t=2\eta^1+2\xi^1$ and $R:=\xi^1-\xi^2+\eta^1-\eta^2=0$.
Then in view of the system~\eqref{eq:RDEsDetEqsForclassCAndLinG}, we have $\tau_{tt}=2\eta^1$,
$R_{tt}=-\xi^2+\eta^1=0$ and $R_{ttt}=\xi^2+\eta^1=0$, i.e., $\xi^2=\eta^1=0$
for Lie-symmetry vector fields of the equation~$\mathfrak E$,
which contradicts the condition $(\xi^2,\eta^1)\ne(0,0)$ for~$Q$.

Therefore, $k=1$.
We fix a template-form equation, and let $(a_1,a_2,a_3,a_4)$ be the corresponding (nonzero) coefficient tuple.
Since $k=1$, the coefficient tuple $(-\xi^2,\eta^2-\xi^1,\eta^1,\tau_t-2\xi^1)$ of the equation~\eqref{eq:RDEsDetEqA}
is proportional to $(a_1,a_2,a_3,a_4)$,
\begin{gather}\label{eq:RDEsFurcateSplitting0}
-\xi^2=\lambda a_1,\quad
\eta^2-\xi^1=\lambda a_2,\quad
\eta^1=\lambda a_3,\quad
\tau_t-2\xi^1=\lambda a_4,
\end{gather}
where $\lambda=\lambda(t)$ is a nonvanishing smooth function of~$t$,
and $(a_1,a_3)\ne(0,0)$ since $(\xi^2,\eta^1)\ne(0,0)$ for~$Q$.
This implies that $a_1\eta^1+a_3\xi^2=0$.
The differentiation of this equation with respect to~$t$ gives,
in view of~\eqref{eq:RDEsDetEqsForclassCAndLinG}, the equation $a_1\eta^1-a_3\xi^2=0$,
i.e., $a_1\eta^1=a_3\xi^2=0$ and thus $\xi^2\eta^1=a_1a_3=0$.

If $a_1\ne0$, then due to the possibility of simultaneous opposite scalings of~$(a_1,a_2,a_3,a_4)$ and~$\lambda$
we can assume without loss of generality that $a_1=1$ and $\lambda=-\xi^2\ne0$.
Hence $\eta^1=0$ and $a_3=0$.
The other equations following from~\eqref{eq:RDEsFurcateSplitting0} are
$\eta^2=-a_2\xi^2+\xi^1$ and $\tau_t=-a_4\xi^2+2\xi^1$.
We combine them with equations of~\eqref{eq:RDEsDetEqsForclassCAndLinG} to get
$\tau_t=\eta^2_t=a_2\xi^2$, which leads to $a_4=-a_2$ and $\xi^1=0$.
We have $a_2\ne0$ since otherwise the fixed template-form equation for~$f$
contradicts the auxiliary inequality $(u_x{}^{\!2}f)_{u_x}\ne0$
for arbitrary elements of equations in the class~$\mathcal C$.
This is why modulo $G^\sim_{\mathcal R}$-equivalence we can set $a_2=1$
and obtain, after alternating the sign of~$t$ if $\varepsilon=-1$, Case~\ref{RDEPeculiarCase}.

For the case $a_3\ne0$, the consideration is similar.
Due to the possibility of simultaneous opposite scalings of~$(a_1,a_2,a_3,a_4)$ and~$\lambda$
we can assume without loss of generality that $a_3=1$ and $\lambda=\eta^1\ne0$.
Hence $\xi^2=0$ and $a_1=0$.
The other equations following from~\eqref{eq:RDEsFurcateSplitting0} are
$\eta^2=a_2\eta^1+\xi^1$ and $\tau_t=a_4\eta^1+2\xi^1$.
We combine them with equations of~\eqref{eq:RDEsDetEqsForclassCAndLinG} to obtain
$\tau_t=\eta^2_t=a_2\eta^1$, which leads to $a_4=a_2$ and $\xi^1=0$.
We have $a_2\ne0$ since otherwise the fixed template-form equation for~$f$
contradicts the auxiliary inequality $f_{u_x}\ne0$
for arbitrary elements of equations in the class~$\mathcal C$.
This is why modulo $G^\sim_{\mathcal R}$-equivalence we can set $a_2=1$
and obtain, after alternating the sign of~$t$ if $\varepsilon=-1$, Case~\ref{RDE:f=1/(u_x+1),g=u}.
\end{proof}

Therefore, in the rest of the proof we can exclude equations
falling, modulo $G^\sim_{\mathcal R}$-equivalence, into Cases~\ref{RDE:f=1/(u_x+1),g=u} or~\ref{RDEPeculiarCase}
and additionally set $\xi_t=\xi_u=\eta_x=0$, i.e.,
\[
\xi=c_1x+c_2,\quad \eta=\eta^2(t)u+\eta^0(t),
\]
where $c_1$ and~$c_2$ are constants.
The unsolved determining equations are exhausted by the reduced form
of the equations~\eqref{eq:RDEsDetEqA} and~\eqref{eq:RDEsSpecificDetEqsForClassC3},
\begin{gather}\label{eq:RDEsReducedDetEqA}
(\eta^2-c_1)u_xf_{u_x}+(\tau_t-2c_1)f=0,
\\ \label{eq:RDEsReducedSpecificDetEqsForClassC3}
(\eta^2u+\eta^0)g_u+(\tau_t-\eta^2)g=\eta^2_tu+\eta^0_t.
\end{gather}

We complete the group classification of the class~$\mathcal C$ by the method of furcate splitting.
The system of the reduced classifying equations~\eqref{eq:RDEsReducedDetEqA} and~\eqref{eq:RDEsReducedSpecificDetEqsForClassC3}
is decoupled, and the arbitrary elements~$f$ and~$g$ are unary functions of different arguments.
Thus, we need to apply two-step furcate splitting, with respect to~$f$ and then with respect to~$g$.
Here the template forms of equations for~$f$ and~$g$ are respectively
\begin{gather}\label{eq:RDEsReducedTemplateEqForF}
a_2u_xf_{u_x}+a_4f=0,
\\ \label{eq:RDEsReducedTemplateEqForG}
(b_1u+b_2)g_u +b_3g=b_4u+b_5,
\end{gather}
where $a_2$, $a_4$ and $b_1$, \dots, $b_5$ are constants.
(We number the constants $a_2$ and $a_4$ in the way consistent with the template form~\eqref{eq:RDEsTemplateEqForF}.)
Template-form equations for~$f$ and~$g$ are obtained by fixing~$t$
in the classifying equations~\eqref{eq:RDEsReducedDetEqA} and~\eqref{eq:RDEsReducedSpecificDetEqsForClassC3}, respectively.
Denote by~$k$ (resp.\ $l$) the number of linearly independent tuples $(a_2,a_4)$ (resp.\ $(b_1,\dots,b_5)$)
among those associated with template-form equations for~$f$ (resp.\ $g$).
We have $k<2$ since $f\ne0$.

\medskip

\noindent $\boldsymbol{k=0}$.
This means that the equation~\eqref{eq:RDEsReducedDetEqA} is an identity with respect to~$f$,
i.e., $\eta^2=c_1$ and $\tau_t=2c_1$.
Then $b_3=b_1$ and $b_4=0$ in the template form~\eqref{eq:RDEsReducedTemplateEqForG} of equations for~$g$.
The number~$l$ of linearly independent tuples among possible values
of the coefficient tuple $(b_1,b_2,b_5)$ cannot exceed one
since otherwise the corresponding system of template-form equations
either is inconsistent or has only constant solutions,
which contradicts the auxiliary inequality $g_u\ne0$ for arbitrary elements of equations in the class~$\mathcal C$.
If $l=0$, i.e., the classifying equation~\eqref{eq:RDEsReducedSpecificDetEqsForClassC3} is also an identity,
then we get the general classification Case~\ref{RDE:f=any,g=any} of Table~\ref{tab:RDEGroupClassification} with no Lie symmetry extension.
For $l=1$, depending on whether or not the coefficient~$b_1$ in a nonidentical template-form equation vanishes,
we obtain that either $g=u\bmod G^\sim_\mathcal R$ or $g=u^{-1}+\nu\bmod G^\sim_{\mathcal R}$.
The first option corresponds to Case~\ref{RDE:f=any,g=u}.
The second option leads to Case~\ref{RDE:f=any,g=1/u} in view of the condition $\nu=0$
since for $\nu\ne0$ we only obtain the kernel invariance algebra $\mathfrak g^\cap_{\mathcal C}=\langle\p_t,\,\p_x\rangle$,
which corresponds to the condition $l=0$ contradicting the supposed condition $l=1$.

\medskip

\noindent $\boldsymbol{k=1.}$
We fix a (nonidentical) template-form equation for~$f$, and let $(a_2,a_4)$ be the corresponding (nonzero) coefficient tuple.
The inequality $f\ne0$ implies $a_2\ne0$, and thus we can set $a_2=1$. Denote $n:=-a_4$.
Then the equation~\eqref{eq:RDEsReducedTemplateEqForF} implies that $f=|u_x|^n\bmod G^\sim_{\mathcal R}$,
where~$n\ne0$ and $n \ne-2$ according to the auxiliary inequalities $f_{u_x}\ne0$ and $(u_x{}^{\!2}f)_{u_x}\ne0$
for arbitrary elements of equations in the class~$\mathcal C$.
Since $k=1$, the coefficient tuple $(\eta^2-c_1,\tau_t-2c_1)$ of the equation~\eqref{eq:RDEsReducedDetEqA}
is proportional to $(1,-n)$.
In other words, $\eta^2-c_1=\lambda$, $\tau_t-2c_1=-n\lambda$,
where $\lambda=\lambda(t)$ is a nonvanishing smooth function of~$t$.
Hence $\tau_t=-n\eta^2+(n+2)c_1$,
and the equation~\eqref{eq:RDEsReducedSpecificDetEqsForClassC3} takes the form
\begin{gather}\label{eq:RDEsReducedSpecificDetEqsForClassC3a}
(\eta^2u+\eta^0)g_u+\big((n+2)c_1-(n+1)\eta^2\big)g=\eta^2_tu+\eta^0_t.
\end{gather}
We again apply the furcate splitting with respect to~$g$ using the number~$l$ for marking different cases.
We have $l>0$ since otherwise the equation~\eqref{eq:RDEsReducedSpecificDetEqsForClassC3a} would be
an identity with respect to~$g$ and thus $\eta^0=\eta^2=0$ and, in view of $n\ne-2$, $c_1=0$,
which would lead to the kernel invariance algebra $\mathfrak g^\cap_{\mathcal C}=\langle\p_t,\,\p_x\rangle$
and would thus give $k=0$, contradicting the condition $k=1$.
The further consideration splits into two cases, $l\geqslant2$ and $l=1$.

\medskip

\noindent $l\geqslant2$. Note that three is the maximal value of~$l$
since for greater values of~$l$, systems of template-form equations are not consistent.

Suppose that for two template-form equations
with linearly independent tuples of coefficients $(b_1,\dots,b_5)$ and $(b_1',\dots,b_5')$,
the conditions
$b_2b_1'-b_2'b_1\ne0$, $b_3b_1'-b_3'b_1=0$ and $b_4b_1'-b_4'b_1\ne0$ hold.
Then the consistency of these template-form equations implies that $g_{uuu}=0$ and $g_{uu}\ne0$,
i.e., $g=u^2+\delta\bmod G^\sim_{\mathcal R}$ with $\delta\in\{-1,0,1\}$.
Splitting the equation~\eqref{eq:RDEsReducedSpecificDetEqsForClassC3a}
for the obtained value of the arbitrary element~$g$ with respect to~$u$,
we derive the equations $(n+2)c_1=(n-1)\eta^2$, $\eta^0_t=-2\delta\eta^2$ and $\eta^2_t=2\eta^0$.
These equations imply $n=1$ since otherwise $\eta^2=(n+2)c_1/(n-1)$, $\eta^0=0$, $\delta c_1=0$
and thus $l\leqslant1$, which contradicts the assumed condition $l\geqslant2$.
Then $c_1=0$, and we obtain Cases~\ref{RDE:f=u_x,g=u^2}a, \ref{RDE:f=u_x,g=u^2}b and~\ref{RDE:f=u_x,g=u^2}c
depending on the value of~$\delta$.
\looseness=-1

If there exists a pair of template-form equations with other conditions on
the associated linearly independent tuples of coefficients $b$'s and $b'$'s,
then we get $g_{uu}=0$, i.e., $g=\varepsilon u\bmod G^\sim_{\mathcal R}$ with $\varepsilon\in\{-1,1\}$
in view of the auxiliary inequality $g_u\ne0$ within the class~$\mathcal C$.
The equation~\eqref{eq:RDEsReducedSpecificDetEqsForClassC3a} with this value of~$g$ splits into
the equations $\eta^2_t=-\varepsilon n\eta^2+\varepsilon(n+2)c_1$ and $\eta^0_t=\varepsilon\eta^0$.
As a result, we get Case~\ref{RDE:f=u_x^k,g=0}b, for which in fact $l=3$.

\medskip

\noindent$l=1.$
Then the coefficient tuple of the equation~\eqref{eq:RDEsReducedSpecificDetEqsForClassC3a} is proportional
to the (nonvanishing) coefficient tuple of a template-form equation,
\[
\eta^2=\chi b_1,\quad
\eta^0=\chi b_2,\quad
(n+2)c_1-(n+1)\eta^2=\chi b_3,\quad
-\eta^2_t=\chi b_4,\quad
-\eta^0_t=\chi b_5.
\]
$\chi=\chi(t)$ is a nonvanishing smooth function of~$t$.
Moreover, we have the inequality $g_{uu}\ne0$ since otherwise $l=2$.
We consider separately two cases, $(b_1,b_4)\ne(0,0)$ and $b_1=b_4=0$.

Suppose at first that $(b_1,b_4)\ne(0,0)$.
We recombine the above equations to the equations $b_1\chi_t+b_4\chi=0$ and $b_2\chi_t+b_5\chi=0$
whose consistency as algebraic equations with respect to $(\chi_t,\chi)$ implies
that the pair $(b_2,b_5)$ is proportional to $(b_1,b_4)$.
Therefore, modulo $G^\sim_{\mathcal R}$-equivalence (more specifically, modulo shifts of~$u$)
we can set $b_2=b_5=0$, and thus $\eta^0=0$.
Then it becomes obvious that $b_1\ne0$
since otherwise $b_3b_4\ne0$, which implies $g_{uu}=0$.
Therefore, due to the possibility of simultaneous opposite scalings of~$(b_1,b_3,b_4)$ and~$\chi$
we can assume without loss of generality that $b_1=1$ and $\chi=\eta^2\ne0$.
Hence $(n+2)c_1=(n+1+b_3)\eta^2$ and $\eta^2_t=-b_4\eta^2$.
For $b_4=0$ we directly obtain Case~\ref{RDE:f=u_x^k,g=u^m}a.
If $b_4\ne0$, then $c_1=0$ and $b_3=-n-1$, which gives Case~\ref{RDE:f=u_x^k,g=u^m}b.

Let $b_1=b_4=0$. Then the inequality $g_{uu}\ne0$ implies $b_2\ne0$ and $b_5=0$.
To derive the last equality, we use the fact that otherwise $\eta^0_t=b_2\chi_t\ne0$ and hence $b_3=0$,
which contradicts the inequality $g_{uu}\ne0$.
Therefore, $g=\varepsilon e^u\bmod G^\sim_{\mathcal R}$ with $\varepsilon\in\{-1,1\}$,
which gives Case~\ref{RDE:f=u_x^k,g=e^u}.

The proof of Proposition~\ref{pro:RDEsGroupClassificationOfC} is completed.

\section{Lie reductions and exact solutions}\label{sec:RDEsSolutions}

Exact solutions of equations for the class~$\mathcal H$ were constructed by various methods in many papers,
including classical Lie reductions~\cite{Dorodnitsyn1982},
nonclassical reductions using conditional symmetries \cite{ArrigoHillBroadbridge1993,ClarksonMansfield1993,FushchichSerov1990,Serov1990,VaneevaPopovychSophocleous2009}
and generalized separation of variables~\cite{GalaktionovSvirshchevskii2007}.
See also~\cite[Section~10.5]{Ames1994V1}, \cite[Chapter~5]{PolyaninZaitsev2012} and references therein.

The construction of exact solutions of equations from the class~$\mathcal F$
is reduced by the family of transformations from~$G^\sim_{\mathcal F}$
with $(t,x,u)$-components $\tilde t=t$, $\tilde x=x$, $\tilde u=u-gt$,
which is parameterized by the arbitrary element~$g$,
to the same problem for the subclass~$\mathcal F'$, cf.\ Section~\ref{sec:RDEEquivalenceTransformations}.
Equations from the subclass~$\mathcal F'$, which are nonlinear filtration equations,
are potential equations of nonlinear diffusion equations.
Lie reductions of the latter equations were studied in~\cite{Ovsiannikov1959}.
Exact solutions of both nonlinear diffusion equations and nonlinear filtration equations were found using
Lie and quasilocal (potential) symmetries \cite{AkhatovGazizovIbragimov1987,AkhatovGazizovIbragimov1991},
generalized separation of variables \cite{BroadbridgeVassiliou2011,Doyle&Vassiliou1998},
nonclassical reductions and nonclassical potential reductions \cite{BlumanYan2005,Gandarias2001,PopovychVaneevaIvanova2007}
and generalized conditional symmetries~\cite{Qu1999}.
See also~\cite[Section~10.2]{Ames1994V1}, \cite[Chapter~5]{PolyaninZaitsev2012} and references therein.

The optimal way for constructing exact solutions of equations from the class~$\mathcal L$
is to use their linearization by the hodograph transformation $\tilde t=t$, $\tilde x=u$, $\tilde u=x$
with~$(\tilde t,\tilde x)$ and~$\tilde u$ being the new independent and dependent variables, respectively,
to Kolmogorov equations; see the introduction.

This is why it is justified to look for exact solutions of equations from the class~$\mathcal C$ only,
excluding equations related to Case~\ref{RDE:f=u_x^k,g=0}b,
which is similar to Case~\ref{RDE:f=u_x^k,g=0}a with respect to additional equivalence transformations.
We select all $\mathcal G^\sim_{\mathcal R}$-inequivalent cases of Table~\ref{tab:RDEGroupClassification}
with four-dimensional maximal Lie invariance algebras as the most interesting,
which are Cases~\ref{RDE:f=1/(u_x+1),g=u}, \ref{RDE:f=u_x,g=u^2}a and~\ref{RDEPeculiarCase}.
They are the obvious contenders for allowing one to construct nontrivial explicit solutions.
We identify the structure of the associated maximal Lie algebras
according Mubarakzyanov's classification of real four-dimensional Lie algebras~\cite{Mubarakzyanov1963a}.
Optimal lists of subalgebras of such algebras
(i.e., complete lists of subalgebras that are inequivalent up to inner automorphisms of the corresponding algebras)
were constructed by Patera and Winternitz in~\cite{PateraWinternitz1977}.
Since we carry out Lie reductions of partial differential equations with two independent variables
to ordinary differential equations, we only need optimal lists of one-dimensional subalgebras.
For each of the selected classification cases, we additionally gauge, whenever it is possible, subalgebra parameters
using outer automorphisms of the corresponding algebras that are induced by discrete symmetries of related equations.
See also~\cite{PopovychBoykoNesterenkoLutfullin2005}
for a collection of various characteristics of low-dimensional Lie algebras
and of classifications of related objects.
It is convenient to present the results of Lie reductions as Table~\ref{tab:RDELieReductions1}, each row of which includes
a basis element of the canonical representative of an equivalence class of one-dimensional subalgebras,
an associated ansatz for the dependent variable, the invariant independent variable~$\omega$ and the corresponding reduced equation, respectively.
Here and below $c$'s are arbitrary constants.

\begin{table}[!t]\footnotesize
\caption{Lie reductions of Cases~\ref{RDE:f=1/(u_x+1),g=u} and~\ref{RDE:f=u_x,g=u^2}a and the equation~\eqref{eq:RDEsPeculiarQuadratic} to ODEs} 
\label{tab:RDELieReductions1}
\begin{center}
\renewcommand{\arraystretch}{1.3}
\setcounter{tbn}{-1}
\begin{tabular}{|c|c|l|c|l|}
\hline
no. &                  Subalgebra basis element                     &\hfil Ansatz for $u$   &  $\omega$              &   \hfil Reduced equation\\
\hline
\ref{RDE:f=1/(u_x+1),g=u}.1      &$\p_x                                  $&$u=\phi                              $&$t                 $&$\phi_\omega=\e\phi                                            $\\
\ref{RDE:f=1/(u_x+1),g=u}.2      &$\p_t+\kappa\p_x                       $&$u=\phi                              $&$x-\kappa t        $&$\phi_{\omega\omega}=-(\kappa\phi_\omega+\e\phi)(\phi_\omega+1)$\\
\ref{RDE:f=1/(u_x+1),g=u}.3      &$e^{\e t}(\p_t+\e(u+x)\p_u)            $&$u=\phi e^{\e t}-x                   $&$x                 $&$\phi_{\omega\omega}=\e\omega\phi_\omega                       $\\
\ref{RDE:f=1/(u_x+1),g=u}.4      &$e^{\e t}(\p_t+\e(u+x)\p_u)+\p_x       $&$u=\phi e^{\e t}-x-e^{-\e t}/(2\e)   $&$x+e^{-\e t}/\e    $&$\phi_{\omega\omega}=\phi_\omega(\e \omega-\phi_\omega)        $\\
\hline
\ref{RDE:f=u_x,g=u^2}.1          &$\p_x                                  $&$u=\phi                              $&$t                 $&$\phi_\omega=\phi^2                                            $\\
\ref{RDE:f=u_x,g=u^2}.2          &$\p_t                                  $&$u=\phi                              $&$x                 $&$\phi_\omega\phi_{\omega\omega}+\phi^2=0                       $\\
\ref{RDE:f=u_x,g=u^2}.3          &$\p_t+\p_x                             $&$u=\phi                              $&$x-t               $&$\phi_\omega\phi_{\omega\omega}+\phi^2+\phi_\omega=0           $\\
\ref{RDE:f=u_x,g=u^2}.4          &$t\p_t-u\p_u+\kappa\p_x                $&$u=\phi/t                            $&$x-\kappa\ln|t|    $&$\phi_\omega\phi_{\omega\omega}+\phi^2+\kappa\phi_\omega+\phi=0$\\
\ref{RDE:f=u_x,g=u^2}.5          &$(t^2{+}1)\p_t-(2tu{+}1)\p_u+\kappa\p_x$&$u=(\phi-t)/(t^2+1)                  $&$x-\kappa\arctan t $&$\phi_\omega\phi_{\omega\omega}+\phi^2+\kappa\phi_\omega+1=0   $\\
\hline
$\tilde{\ref{RDEPeculiarCase}}$.1&$\p_x                                  $&$u=\phi                              $&$t                 $&$\phi_\omega=0                                                 $\\
$\tilde{\ref{RDEPeculiarCase}}$.2&$t\p_x+\e^{-1}\p_u                     $&$u=\phi+\e^{-1}x/t                   $&$t                 $&$\phi_\omega+\omega^{-1}\phi=0                                 $\\
$\tilde{\ref{RDEPeculiarCase}}$.3&$\p_t                                  $&$u=\phi                              $&$x                 $&$\phi_\omega(\phi_{\omega\omega}-\e\phi)=0                     $\\
$\tilde{\ref{RDEPeculiarCase}}$.4&$\p_t+t\p_x+\e^{-1}\p_u                $&$u=\phi+\e^{-1}t                     $&$x-t^2/2           $&$\phi_\omega(\phi_{\omega\omega}-\e\phi)=\e^{-1}               $\\
$\tilde{\ref{RDEPeculiarCase}}$.5&$t\p_t-u\p_u+\kappa\p_x                $&$u=(\phi+\e^{-1}\kappa)/t            $&$x-\kappa\ln|t|    $&$\phi_\omega(\phi_{\omega\omega}-\e\phi)+\phi=-\e^{-1}\kappa   $\\
\hline
\end{tabular}
\end{center}
\end{table}

\medskip\par\noindent%
\textbf{Case~\ref{RDE:f=1/(u_x+1),g=u}.}
The maximal Lie invariance algebra~$\mathfrak g_{\ref{RDE:f=1/(u_x+1),g=u}}$ of the equation
\begin{gather}\label{eq:RDEs:f=1/(u_x+1),g=u}
u_t=\frac{u_{xx}}{u_x+1}+\e u, \quad \e=\pm1\bmod G^\sim_{\mathcal R},
\end{gather}
is isomorphic to the algebra~$\textup g_{4,8}$ with $h=0$ from Mubarakzyanov's classification,
which is denoted by~$A^0_{4.8}$ in~\cite{PopovychBoykoNesterenkoLutfullin2005}
and by~$A^0_{4,9}$ in~\cite{PateraWinternitz1977}.
To obtain canonical commutation relations, basis elements can be chosen in the following way:
\[
e_1=e^{\e t}\p_u,\quad
e_2=e^{\e t}(\p_t+\e(u+x)\p_u),\quad
e_3=\e^{-1}\p_x,\quad
e_4=\e^{-1}\p_t.
\]
An optimal list of one-dimensional subalgebras of~$\mathfrak g_{\ref{RDE:f=1/(u_x+1),g=u}}$ is exhausted by the subalgebras
$\langle e_1\rangle$, $\langle e_2\rangle$, $\langle e_3\rangle$, $\langle e_2+\e'\e e_3\rangle$ and $\langle e_4+\kappa e_3\rangle$,
where $\e'=\pm1$ and $\kappa$ is an arbitrary constant.
The equation~\eqref{eq:RDEs:f=1/(u_x+1),g=u} also admits
the discrete point symmetry~$I_{x,u}$ of simultaneously alternating signs of~$(x,u)$,
which generates an outer automorphism of the algebra~$\mathfrak g_{\ref{RDE:f=1/(u_x+1),g=u}}$
with the matrix ${\rm diag}(-1,1,-1,1)$ in the chosen basis.
Using this automorphism we can set $\e'=1$ in the fourth subalgebra.
The first subalgebra does not satisfy the transversality condition and thus cannot be used for Lie reduction.
Ansatzes constructed with the other subalgebras and the corresponding reduced equations are listed in Table~\ref{tab:RDELieReductions1}.
The general solutions of reduced equations~\ref{RDE:f=1/(u_x+1),g=u}.1, \ref{RDE:f=1/(u_x+1),g=u}.3 and~\ref{RDE:f=1/(u_x+1),g=u}.4
are respectively
\[\textstyle
\phi=c_1e^{\e t},\quad
\phi=c_1\int e^{\varepsilon\omega^2/2}\mathrm d\omega+c_2,\quad
\phi=\ln\big|\int e^{\varepsilon\omega^2/2}\mathrm d\omega+c_1\big|+c_2,
\]
where $c$'s are arbitrary constants, and integrals denote fixed antiderivatives.
The reduced equation~\ref{RDE:f=1/(u_x+1),g=u}.2 with $\kappa=0$ is integrated once to $\phi_\omega-\ln|\phi_\omega+1|=-\e\phi^2/2+c_1$.
The general solution of the last equation is represented in a parametric form with quadrature,
where $\phi_\omega$ plays the role of the parameter.

\par\medskip\noindent\textbf{Case~\ref{RDE:f=u_x,g=u^2}a.}
The maximal Lie invariance algebra~$\mathfrak g_{\ref{RDE:f=u_x,g=u^2}a}$ of the equation $u_t=u_xu_{xx}+u^2$
is a realization of the algebra~$\textup g_{3,5}+\textup g_1$ from Mubarakzyanov's classification,
which is denoted by~$A_{3,8}\oplus A_1$ in~\cite{PateraWinternitz1977}
and is merely ${\rm sl}(2,\mathbb R)\oplus A_1$~\cite{PopovychBoykoNesterenkoLutfullin2005}.
We choose the basis
\[
e_1=\p_t,\quad
e_2=t\p_t-u\p_u,\quad
e_3=t^2\p_t-(2tu+1)\p_u,\quad
e_4=\p_x.
\]
An optimal list of one-dimensional subalgebras of~$\mathfrak g_{\ref{RDE:f=u_x,g=u^2}a}$ is exhausted by the subalgebras
$\langle e_1\rangle$, $\langle e_4\rangle$, $\langle e_2+\e'e_4\rangle$, $\langle e_2+\kappa e_4\rangle$ and $\langle e_1+e_3+\kappa e_4\rangle$,
where $\e'=\pm1$, and $\kappa$ is an arbitrary nonnegative constant and an arbitrary constant in the fourth and the fifth subalgebras, respectively.
The equation $u_t=u_xu_{xx}+u^2$ also admits
the discrete point symmetry~$I_{t,u}$ of simultaneously alternating signs of~$(t,u)$,
which generates the outer automorphism of the algebra~$\mathfrak g_{\ref{RDE:f=u_x,g=u^2}a}$
with the matrix ${\rm diag}(-1,1,-1,1)$ in the chosen basis.
Using this automorphism we can set $\e'=1$ in the third subalgebra.
Ansatzes constructed with the above subalgebras and the corresponding reduced equations are listed in Table~\ref{tab:RDELieReductions1}.

The solutions of reduced equation~\ref{RDE:f=u_x,g=u^2}.1 are exhausted by
the general solution $\phi=-(t+c_1)^{-1}$ and the singular solution $\phi=0$.

Reduced equation~\ref{RDE:f=u_x,g=u^2}.2 is integrated to $-\int(\phi^3+c_1)^{-1/3}{\rm d}\phi=x+c_2$.
For the value $c_1=0$, we obtain the one-parameter singular solution family $\phi=\tilde c_2e^{-x}$.
If $c_1\ne0$, then the expression $(\phi^3+c_1)^{-1/3}{\rm d}\phi$ is a differential binomial $\phi^m(a+b\phi^n)^p{\rm d}\phi$
with $m=0$, $n=3$, $p=-1/3$, $a=c_1$ and $b=1$, i.e., $(m+1)/n+p=0$ is an integer.
According the Chebyshev theorem on the integration of differential binomials,
the integral of this differential binomial is reduced to the integral of a rational function by the substitution
$c_1\phi^{-3}+1=t^3$ and thus it can be expressed by elementary functions.
Finally, we construct the general solution of reduced equation~\ref{RDE:f=u_x,g=u^2}.2 in the implicit form
\[
\frac12\ln|(\phi^3+c_1)^{1/3}-\phi|+\frac1{\sqrt3}\arctan\frac{2(\phi^3+c_1)^{1/3}+\phi}{\sqrt3\phi}=x+c_2.
\]

Reduced equations~\ref{RDE:f=u_x,g=u^2}.3--\ref{RDE:f=u_x,g=u^2}.5 are of the general form
\begin{gather}\label{eq:RDECase6aGenFormOfReducedEqs}
\phi_\omega\phi_{\omega\omega}+\phi^2+\kappa\phi\phi_\omega+\mu\phi+\nu=0,
\end{gather}
where $(\kappa,\mu,\nu)=(1,0,0)$, $(\mu,\nu)=(1,0)$ and $(\mu,\nu)=(0,1)$
for reduced equation~\ref{RDE:f=u_x,g=u^2}.3, \ref{RDE:f=u_x,g=u^2}.4 and~\ref{RDE:f=u_x,g=u^2}.5.
By the standard substitution lowering orders of autonomous equations,
$\phi_\omega=p(y)$, where $y=\phi$ plays the role of the new independent variable,
the equation~\eqref{eq:RDECase6aGenFormOfReducedEqs} is reduced to the first-order ordinary differential equation
$p^2p_y+y^2+\kappa p+\mu y+\nu=0$.
The point transformation $\tilde y=p+y$, $\tilde p=y$ maps the last equation
to an Abel equation of the second kind,
\[\big((2\tilde y+\mu-\kappa)\tilde p-\tilde y^2+\kappa\tilde y+\nu\big)\tilde p_{\tilde y}+(\tilde p-\tilde y)^2=0,\]
see, e.g., \cite[substitution~1.4.4-1.3$^\circ$]{PolyaninZaitsev2003}.
If $\nu=0$, one can also use the point transformation $\tilde y=p/y$, $\tilde p=1/y$ resulting to
the simpler Abel equation of the second kind
$\big((\kappa\tilde y+\mu)\tilde p+\tilde y^3+1\big)\tilde p_{\tilde y}=\tilde y^2\tilde p.$
The above Abel equations are further reduced by point transformations
to Abel equations of the second (or first) kind in the canonical form.

\par\medskip\noindent\textbf{Case~\ref{RDEPeculiarCase}.}
Up to $G^\sim_{\mathcal R}$-equivalence, this is the only case
among cases of Lie-symmetry extensions within the subclass~$\mathcal C$
that admits non-fiber-preserving point symmetry transformations.
(Such symmetry transformations are typical for equations from the subclasses~$\mathcal F$ and~$\mathcal L$.)
The indication, in infinitesimal terms, of the presence of such symmetry transformations is that
the $x$-component of the Lie-symmetry vector field $e^{-\varepsilon t}(\p_t-\varepsilon u\p_x+\varepsilon u\p_u)$
of this case depends on~$u$.
To avoid Lie reductions with implicit ansatzes and to simplify the Lie reduction procedure in total,
it is convenient to map the corresponding equation
\begin{gather}\label{eq:RDEsPeculiarEq}
u_t=\frac{u_xu_{xx}}{(u_x+1)^3}+\e u, \quad \e=\pm1\bmod G^\sim_{\mathcal R},
\end{gather}
by the point transformation $\tilde t=\e^{-1}e^{\e t}$, $\tilde x=x+u$, $\tilde u=e^{-\e t}u$
to the equation
\begin{gather}\label{eq:RDEsPeculiarQuadratic}
{\tilde u}_{\tilde t}=\tilde u_{\tilde x}\tilde u_{\tilde x\tilde x}-\e\tilde u\tilde u_{\tilde x},
\end{gather}
which we denote by ``Case~$\tilde{\ref{RDEPeculiarCase}}$''.
We omit tildes of $t$, $x$ and~$u$ below.
The maximal Lie invariance algebra~$\mathfrak g_{\tilde{\ref{RDEPeculiarCase}}}$ of the equation~\eqref{eq:RDEsPeculiarQuadratic}
is spanned by the vector fields $e_1=\p_x$, $e_2=\p_t$, $e_3=t\p_x+\e^{-1}\p_u$ and $e_4=t\p_t-u\p_u$.
It is a realization of the algebra~$\textup g_{4,8}$ with $h=-1$ from Mubarakzyanov's classification,
which is denoted by~$A_{4,8}$ in~\cite{PateraWinternitz1977}
and by $A_{4.8}^{-1}$ in~\cite{PopovychBoykoNesterenkoLutfullin2005}.
An optimal list of one-dimensional subalgebras of~$\mathfrak g_{\tilde{\ref{RDEPeculiarCase}}}$ is exhausted by the subalgebras
$\langle e_1\rangle$, $\langle e_3\rangle$, $\langle e_2\rangle$, $\langle e_2+\e'e_3\rangle$ and $\langle e_4+\kappa e_1\rangle$,
where $\e'=\pm1$ and $\kappa$ is an arbitrary constant.
The equation~\eqref{eq:RDEsPeculiarQuadratic}  admits
the discrete point symmetries~$I_{t,u}$ and~$I_{x,u}$ of simultaneously alternating signs of~$(t,u)$ and of~$(x,u)$, respectively.
The involution~$I_{x,u}$ generates the outer automorphism of the algebra~$\mathfrak g_{\tilde{\ref{RDEPeculiarCase}}}$
with the matrix ${\rm diag}(-1,1,-1,1)$ in the fixed basis.
Using this automorphism we can set $\e'=1$ in the fourth subalgebra.
Ansatzes constructed with the above subalgebras and the corresponding reduced equations
are marked in Table~\ref{tab:RDELieReductions1} by~$\tilde{\ref{RDEPeculiarCase}}$,
where tildes of $t$, $x$ and~$u$ are still omitted.

Reduced equations~$\tilde{\ref{RDEPeculiarCase}}$.1 and~$\tilde{\ref{RDEPeculiarCase}}$.2
are trivially integrated to $\phi=c_0$ and $\phi=c_0/t$, respectively.

The solution set of reduced equation~$\tilde{\ref{RDEPeculiarCase}}$.3 splits into two families.
The first family $\phi=c_0$ is common for both the values of~$\e$.
The second family depends on~$\e$,
$\phi=c_1e^\omega+c_2e^{-\omega}$ if $\e=1$ and
$\phi=c_1\cos\omega+c_2\sin\omega$ if $\e=-1$.

Reduced equation~$\tilde{\ref{RDEPeculiarCase}}$.4 has the first integral $\phi_\omega^{\,2}-\e\phi^2-2\e^{-1}\omega=c_1$.
This equation is mapped by the substitution
$1/\phi_\omega+\phi=p(y)$, where $y=-1/\phi_\omega$ plays the role of the new independent variable,
to the Abel equation of the second kind in the canonical form $pp_y-p=\e/y^3$.

Similarly, reduced equation~$\tilde{\ref{RDEPeculiarCase}}$.5 is mapped by the substitution
$p(y)=\phi$, where $y=\phi_\omega$ plays the role of the new independent variable,
to the equation $((1-\e y)p+\e^{-1}\kappa)p_y=-y^2$,
which is an Abel equation of the second kind if $\kappa\ne0$ and
an equation with separable variables if $\kappa=0$.
In the second case, we obtain the equation
$
\phi^2=\e^{-1}\phi_\omega^{\,2}+2\e^{-2}\phi_\omega+2\e^{-3}\ln|\phi_\omega-\e^{-1}|+c_1
$
whose general solution is represented in a parametric form with quadrature,
where $\phi_\omega$ plays the role of the parameter.
If $\kappa\ne0$, the corresponding Abel equation is reduced to the canonical form
$\bar p\bar p_{\bar y}-\bar p=\e^{-3}(\bar y+\e^{-1}\kappa)^2/\bar y^3$
by the point transformation $\bar y=\e^{-1}\kappa/(1-\e y)$, $\bar p=p+\e^{-1}\kappa/(1-\e y)$.

The equation~\eqref{eq:RDEsPeculiarQuadratic} is of the form
admitting the generalized separation of variables~\cite{GalaktionovPosashkovSvirshchevskii1995}.
The nonlinear differential operator $\tilde u_{\tilde x}\tilde u_{\tilde x\tilde x}-\e\tilde u\tilde u_{\tilde x}$
from the right hand side of this equation preserves
the exponential linear space $\langle1,e^{\tilde x},e^{-\tilde x}\rangle$
or trigonometric linear spaces $\langle1,\cos\tilde x,\sin\tilde x\rangle$
if $\e=1$ or $\e=-1$, respectively, cf.~\cite{GalaktionovSvirshchevskii2007}.
Using the ansatz $\tilde u=\tau^0(\tilde t)+\tau^1(\tilde t)e^{\tilde x}+\tau^2(\tilde t)e^{-\tilde x}$ for $\e=1$
or $\tilde u=\tau^0(\tilde t)+\tau^1(\tilde t)\cos\tilde x+\tau^2(\tilde t)\sin\tilde x$ for $\e=-1$,
where $\tau$'s are the new unknown functions, we
construct three-parameter families of exact solutions of the equation~\eqref{eq:RDEsPeculiarQuadratic},
\[
\tilde u=c_0+c_1e^{\tilde x+c_0\tilde t}+c_2e^{-\tilde x-c_0\tilde t}\quad\mbox{if}\quad\e=1\quad\mbox{and}\quad
\tilde u=c_0+c_1\cos(\tilde x+c_0\tilde t+c_2)\quad\mbox{if}\quad\e=-1.
\]
At the same time, these solutions are Lie invariant solutions,
which can be constructed using reduction~$\tilde{\ref{RDEPeculiarCase}}$.3
and the action of Galilean boosts generated by the Lie-symmetry vector field~$e_3$.
No further extension of these families of exact solutions of the equation~\eqref{eq:RDEsPeculiarQuadratic}
with its Lie symmetries is possible.

Substituting $\tilde t=\e^{-1}e^{\e t}$, $\tilde x=x+u$, $\tilde u=e^{-\e t}u$ into
the above solutions of the equation~\eqref{eq:RDEsPeculiarQuadratic},
we construct implicit solutions of the equation~\eqref{eq:RDEsPeculiarEq}.

Note that the further hodograph transformation $\hat t=\tilde t$, $\hat x=\tilde u$, $\hat u=\tilde x$
reduces the equation~\eqref{eq:RDEsPeculiarQuadratic}
to the equation $\hat u_{\hat t}=\hat u_{\hat x}^{-3}\hat u_{\hat x\hat x}+\e x$,
which is the potential equation for the nonlinear diffusion--reaction equations
$\breve u_{\hat t}=(\breve u^{-3}\breve u_{\hat x})_{\hat x}+\e$, where $\breve u=\hat u_{\hat x}$.
See the first paragraph of this section for references on symmetry analysis of the last equation.

\section{Conclusion}\label{sec:RDEsConclusions}

Carrying out group classification of the class~$\mathcal R$ of (1+1)-dimensional nonlinear diffusion--reaction equations
with gradient-dependent diffusivity in the present paper,
we have corrected and enhanced results of~\cite{ChernihaKingKovalenko2016} in several aspects.

We have justified the necessity of studying the entire class~$\mathcal R$ and carefully analyzed its structure.
In the notation of the present paper, the authors of~\cite{ChernihaKingKovalenko2016}
considered the group classification problem only for the subclass~$\mathcal C$ with respect to its equivalence group,
having discarded other subclasses of~$\mathcal R$.
They include the subclass~$\mathcal H$ with known group classification,
the subclass~$\mathcal F$ whose group classification can be easily derived from known group classification of the subclass~$\mathcal F'$
and the subclass~$\mathcal L$ of linearizable equations with transformational properties
essentially different from other equations in the class~$\mathcal R$.
Nevertheless, the consideration of the subclass~$\mathcal C$ alone is not natural
since some equations in~$\mathcal C$ are related by point transformations to equations in~$\mathcal F$,
and the subclass~$\mathcal F$ intersects the subclasses~$\mathcal H$ and~$\mathcal L$.

We have computed the usual equivalence groups of the class~$\mathcal R$ and all the above subclasses
as well as the generalized equivalence group~$\bar G^\sim_{\mathcal F}$
and the effective generalized equivalence group~\smash{$\hat G^\sim_{\mathcal F}$} of the subclass~$\mathcal F$.
We have also checked the consistency of these groups for using in the course of group classification.
The group~$\hat G^\sim_{\mathcal F}$ gives the first nontrivial example of \emph{finite-dimensional effective generalized equivalence group}
in the literature.
Moreover, the class~$\mathcal F$ has another unexpected property formulated in Theorem~\ref{thm:RDEsOnContainingUsualEquivGroupOfF}:
any effective generalized equivalence group of the class~$\mathcal F$
does not contain the usual equivalence group of this class.
This phenomenon had not been observed before, and finding out it is the most interesting result of the present paper
although it was obtained as a by-product.
Since~\smash{$\hat G^\sim_{\mathcal F}$} is not a normal subgroup of~$\bar G^\sim_{\mathcal F}$,
it is obvious that \smash{$\hat G^\sim_{\mathcal F}$} is not a unique effective generalized equivalence group of the subclass~$\mathcal F$.
Whether this group is unique up to the subgroup similarity within~$\bar G^\sim_{\mathcal F}$ is still an open problem.

In the context of the group classification of the subclass~$\mathcal C$ given in~\cite{ChernihaKingKovalenko2016},
the most significant and also unexpected enhancement
is discovering Case~\ref{RDEPeculiarCase}, which was missed in~\cite{ChernihaKingKovalenko2016}.
It essentially differs from the other classification cases in the subclass~$\mathcal C$
since only equations fitting in Case~\ref{RDEPeculiarCase} up to $G^\sim_{\mathcal R}$-equivalence
admits non-fiber-preserving Lie-symmetry transformations.
An explanation of the presence of such transformations is
that these equations are reduced by non-fiber-preserving point transformations to a potential nonlinear diffusion--reaction equation.

The third direction of enhancements concerns additional equivalences between classification cases.
We have constructed additional equivalence transformations relating Cases~\ref{RDE:f=u_x,g=u^2}a--\ref{RDE:f=u_x,g=u^2}c,
which were missed in~\cite{ChernihaKingKovalenko2016}.
It is obvious that there also exist similar transformations
between multidimensional counterparts of Cases~\ref{RDE:f=u_x,g=u^2}a--\ref{RDE:f=u_x,g=u^2}c,
cf.\ Cases~6--8 of~\cite[Table~2]{ChernihaKingKovalenko2016}.
Using relevant equivalence relations, we have properly gauged all constants parameterizing classification cases
and first proved the completeness of the additional equivalence transformations presented in the footnote of Table~\ref{tab:RDEGroupClassification},
which leads to the group classification of the class~$\mathcal R$ up to the general point equivalence,
see Theorem~\ref{thm:RDEGroupClassificationUpToGenPointEquiv}.

The complex structure of the class~$\mathcal R$ required the application of various advanced techniques
of group analysis of differential equations.
Moreover, we needed to combine and develop them for an efficient solution of the group classification problem for the class~$\mathcal R$.
At the same time, the simple form of equations in this class has allowed us to present the obtained results in a clear~way.

Results of Section~\ref{sec:RDEEquivalenceTransformations} on various equivalence groups
of the class~$\mathcal R$ and of its subclasses~$\mathcal H$, $\mathcal L$, $\mathcal F$, $\mathcal F'$ and~$\mathcal C$
and the classification of Lie symmetries of equations from these classes
that is presented in Section~\ref{sec:RDEsClassGroupClassification}
jointly imply that these classes are not normalized in both the usual and the generalized sense.
(See~\cite{OpanasenkoBihloPopovych2017,Popovych2006,PopovychKunzingerEshraghi2010} for related definitions.)
The existence of additional equivalence transformations among cases of Lie-symmetries extension
in the subclasses~$\mathcal H$, $\mathcal L$ and~$\mathcal C$ means that
these subclasses as well as the entire class~$\mathcal R$ are even not semi-normalized.
This is why in the course of computing the above equivalence groups,
we have suggested and applied an optimized version of the direct method.
This version involves preliminary study of admissible transformations within the entire class
and the successive splitting of the determining equations for these transformations
with respect to the corresponding arbitrary elements and their derivatives,
depending on auxiliary constraints associated with each of the subclasses under consideration separately.
In order to make the group classifications of the subclasses~$\mathcal F$ and~$\mathcal F'$ consistent,
we have found the generalized equivalence group of~$\mathcal F$ and its effective counterpart.
The group classification of the class~$\mathcal L$ has been obtained
from the known group classification of Kolmogorov equations
up to general point equivalence~\cite[Corollary~7]{PopovychKunzingerIvanova2008}
using a technique based on mappings between classes of differential equations
via point transformations~\cite{VaneevaPopovychSophocleous2009}.

The arbitrary elements~$f$ and~$g$ depend on different arguments, $u_x$ and~$u$, respectively.
This is why in the course of solving the group classification problem for the subclass~$\mathcal C$,
we needed to use the double furcate splitting with respect to two pairs ``(an argument, an arbitrary element)'',
$(u_x,f)$ and $(u,g)$.
Moreover, it is impossible to directly apply furcate splitting to the classifying system~\eqref{eq:RDEsDetEqSystem}
within the subclass~$\mathcal C$, not to mention the entire class~$\mathcal R$.
The subsystem of non-classifying determining equations for Lie symmetries of equations from the class~$\mathcal R$
is exhausted by two equations,
which are common for all (1+1)-dimensional evolution equations
and constrain only the $t$-components of Lie-symmetry vector fields, $\tau_x=\tau_u=0$.
The dependence of the corresponding $x$- and $u$-components on~$u$ is not specified
and hence the furcate splitting with respect to $(u,g)$ cannot be initiated
without additional constraints on the arbitrary elements and without a preliminary preparation of the classifying equations. 
One more complication is that the classifying system~\eqref{eq:RDEsDetEqSystem} is (partially) coupled in~$(f,g)$.
Under the auxiliary inequalities $f_{u_x}\ne0$, $(u_x{}^{\!2}f)_{u_x}\ne0$ and $g_u\ne0$
singling out the subclass~$\mathcal C$,
the system~\eqref{eq:RDEsDetEqSystem} implies more equations,
\eqref{eq:RDEsSpecificDetEqsForClassC1}, not involving the arbitrary elements~$(f,g)$
and thus reduces to an uncoupled system for~$(f,g)$; see Lemma~\ref{lem:RDEsDetEqsForC} and the discussion after it.
An inconvenience of the reduced system is that it includes two equations for~$g$,
\eqref{eq:RDEsSpecificDetEqsForClassC2} and~\eqref{eq:RDEsSpecificDetEqsForClassC3}.
Then we have found out a property of Lie symmetry algebras that singles out
two classification cases being singular in the subclass~$\mathcal C$, Cases~\ref{RDE:f=1/(u_x+1),g=u} and~\ref{RDEPeculiarCase}.
In the course of computing these cases, we use the auxiliary furcate spitting with respect $(u_x,f)$.
For other classification cases, there are very restrictive determining equations, $\xi_t=\xi_u=\eta_x=0$,
which simplify the system of classifying equations
to the uncoupled systems of only two equations~\eqref{eq:RDEsReducedDetEqA} and~\eqref{eq:RDEsReducedSpecificDetEqsForClassC3}.
The last system is perfectly appropriate for the double furcate splitting with respect to $(u_x,f)$ and $(u,g)$.

For effectively constructing additional equivalence transformations
or for proving nonexistence of such transformations,
we have combined the direct method of computing point transformations between similar equations
with the algebraic method based on comparing the structure of the corresponding Lie invariance algebras.

\section*{Acknowledgments}

The research of SO was undertaken, in part, thanks to funding from the Canada Research Chairs
program and the NSERC Discovery Grant program.
The authors acknowledge the partial financial support provided by the NAS of Ukraine under the project 0116U003059.
The research of ROP was supported by the Austrian Science Fund (FWF), projects P25064 and P30233.
ROP is also grateful to the project No.\ CZ.$02.2.69\/0.0/0.0/16\_027/0008521$
``Support of International Mobility of Researchers at SU'' which supports international cooperation.

\footnotesize

\end{document}